\documentclass[sigconf, nonacm]{acmart}
\usepackage{graphicx}
\usepackage{amsmath,amsthm}
\usepackage{algorithm}
\usepackage{algpseudocode}
\usepackage{multirow}
\usepackage{subfiles}
\usepackage{multicol}
\usepackage{makecell}
\usepackage{bm}
\usepackage{enumitem}
\newcommand\vldbdoi{XX.XX/XXX.XX}

\newcommand\vldbauthors{\authors}
\newcommand\vldbtitle{\shorttitle} 
\newcommand\vldbavailabilityurl{URL_TO_YOUR_ARTIFACTS}
\newcommand\vldbpagestyle{plain}

\def\ojoin{\setbox0=\hbox{$\bowtie$}%
  \rule[-.02ex]{.25em}{.4pt}\llap{\rule[\ht0]{.25em}{.4pt}}}
\def\leftouterjoin{\mathbin{\ojoin\mkern-5.8mu\bowtie}}

\usepackage{listings}
\usepackage{xcolor}

\lstdefinelanguage{PACSQL}{
  morekeywords={
    CREATE, ALTER, TABLE, PAC, PROTECTED, LINK,
    ADD, DROP, REFERENCES
  },
  sensitive=true,
  morecomment=[l]{--},
  morestring=[b]'
}

\usepackage{style}

\lstdefinestyle{cppstyle}{
  language=C++,
  basicstyle=\ttfamily\footnotesize,
  keywordstyle=\bfseries\color{blue!70!black},
  commentstyle=\color{green!50!black},
  stringstyle=\color{red!60!black},
  numberstyle=\tiny\color{gray},
  numbers=left,
  frame=single,
  backgroundcolor=\color{blue!5},
  breaklines=true,
  tabsize=2,
  morekeywords={uint64_t, void}
}

\begin{document}
\title{SIMD-PAC-DB: Pretty Performant PAC Privacy}

\settopmatter{authorsperrow=4}

\author{Ilaria Battiston}
\affiliation{%
  \institution{CWI}
  \country{The Netherlands}
}
\email{ilaria@cwi.nl}

\author{Dandan Yuan}
\affiliation{%
  \institution{CWI}
  \country{The Netherlands}
}
\email{dandan@cwi.nl}

\author{Xiaochen Zhu}
\affiliation{%
  \institution{MIT}
  \country{USA}
}
\email{xczhu@mit.edu}

\author{Peter Boncz}
\affiliation{%
  \institution{CWI}
  \country{The Netherlands}
}
\email{boncz@cwi.nl}

\begin{abstract}
This work presents a highly optimized implementation of PAC-DB, a recent and promising database privacy model. We prove that our SIMD-PAC-DB can compute the same privatized answer with just a single query, instead of the 128 stochastic executions against different 50\% database sub-samples needed by the original PAC-DB.
Our key insight is that every bit of a hashed primary key can be seen to represent membership of such a sub-sample. We present new algorithms for approximate computation of stochastic aggregates based on these hashes, which, thanks to their SIMD-friendliness, run up to 40x faster than scalar equivalents. We release an open-source DuckDB community extension which includes a rewriter that PAC-privatizes arbitrary SQL queries. Our experiments on TPC-H, Clickbench, and SQLStorm evaluate thousands of queries in terms of performance and utility, significantly advancing the ease of use and functionality of privacy-aware data systems in practice.
\end{abstract}

\maketitle

\pagestyle{\vldbpagestyle}
\begingroup\small\noindent\raggedright\textbf{Citation Format:}\\
\vldbauthors. \vldbtitle. 
\href{https://arxiv.org/abs/2603.15023}{https://arxiv.org/abs/2603.15023}, 2026.
\endgroup
\begingroup
\renewcommand\thefootnote{}\footnote{\noindent
This work is licensed under the Creative Commons BY-NC-ND 4.0 International License. Visit \url{https://creativecommons.org/licenses/by-nc-nd/4.0/} to view a copy of this license. For any use beyond those covered by this license, obtain permission by emailing \href{mailto:ilaria@cwi.nl}{ilaria@cwi.nl}. Copyright is held by the authors. \\
\href{https://arxiv.org/abs/2603.15023}{https://arxiv.org/abs/2603.15023}, 2026. \\
}\addtocounter{footnote}{-1}\endgroup

\ifdefempty{\vldbavailabilityurl}{}{
\vspace{.3cm}
\begingroup\small\noindent\raggedright\textbf{Artifact Availability:}\\
The source code, data, and/or other artifacts have been made available at \url{github.com/cwida/pac}.
\endgroup
}

\section{Introduction}
\label{sec:introduction}
\begin{figure}
  \centering\vspace*{2mm}
  \includegraphics[width=\linewidth]{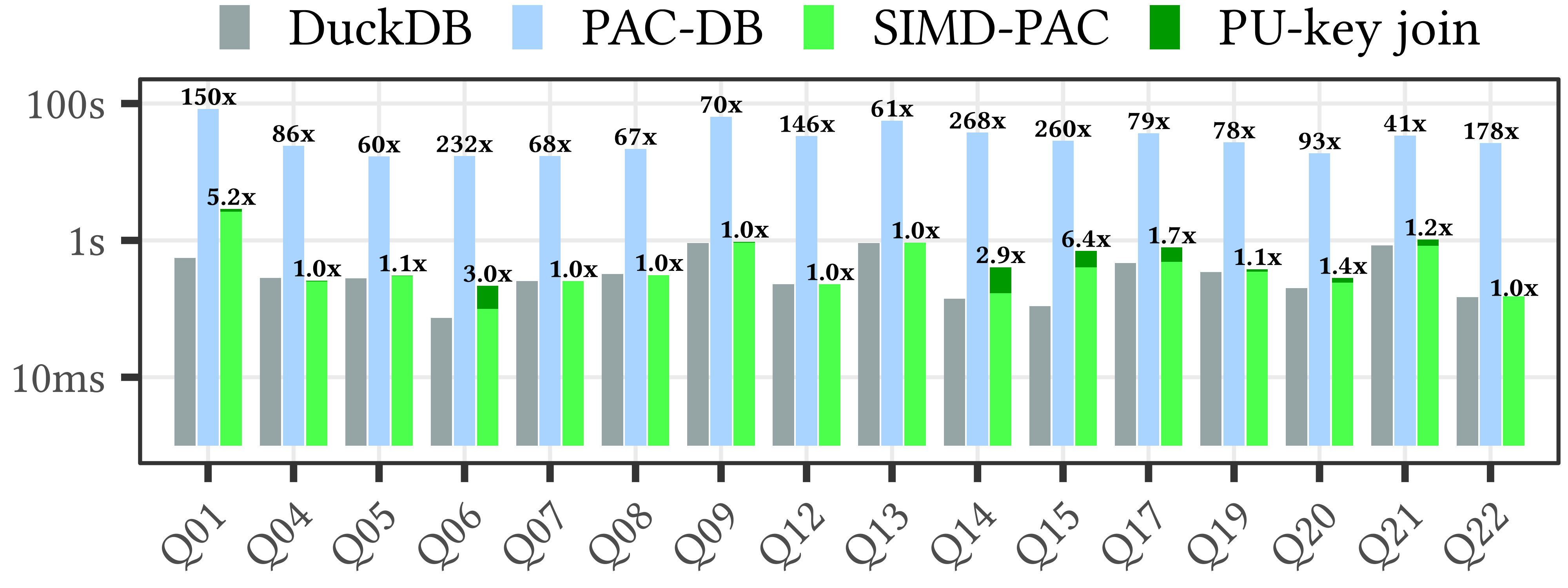}\vspace*{-4mm}\caption{Our SIMD-friendly PAC (light green) incurs limited overhead compared to default DuckDB (TPC-H SF30, Macbook M2 Pro), making it much more usable than PAC-DB (light blue). Slowdowns mostly stem from unavoidable extra joins (dark green) -- aggregation-heavy Q01 being the only exception. PAC rejects Q10 and 18, which release customer data. Q02, 03, 16 are also not shown, as PAC is a no-op there.}\vspace*{-5mm} 
  \label{fig:tpch}
\end{figure}

Recent work on Probably Approximately Correct (PAC) Privacy~\cite{hanshen2023crypto,sridhar2025pac,zhu2026pacadversarial} introduced black-box privatization of arbitrary mechanisms by automatically estimating the noise required to meet a given privacy budget. \textbf{PAC-DB} adapts this concept to privacy-aware query processing~\cite{sridhar2026pacdb}, protecting against membership inference attacks in query sessions. 
PAC-DB has promising applicability in our field, because it can (i) privatize queries automatically and (ii) provide answers with good utility, as it exploits result-specific value distributions, rather than dataset-wide domain bounds, to generate the required noise. 
In contrast, query processing with Differential Privacy (DP) needs trusted curators, as it requires detailed white-box sensitivity analysis of each query. PAC-DB shifts the burden from human analysis to computational effort, enabling privacy guarantees for (almost) arbitrary queries.
The PAC-DB framework\footnote{Original PAC-DB implementation: github.com/michaelnoguera/pacdb} executes a query $m=128$ times on $m$ database subsets (``possible worlds''), each obtained by a uniform 50\% sampling of the rows of the {\em Privacy Unit} (PU), e.g. \texttt{customer} table. This sampling is extended to non-cyclic database schemas towards {\em dependent tables}, which can reach the PU table via foreign key joins, e.g., \texttt{orders}, \texttt{lineitem} in TPC-H. PAC-DB observes the value distribution from collecting the $m$ query results and applies noise based on its variance. Consequently, its runtime cost is $O(m)$ larger than default query execution --- as also can be seen in Figure~\ref{fig:tpch} that shows default, PAC-DB and SIMD-PAC-DB on TPC-H scale factor 30.

We present \textbf{DuckDB-PAC}, our open-source DuckDB extension that executes SQL queries under PAC privacy with overheads of $\sim$2x --- at least one order of magnitude faster. It consists of a query compiler that automatically translates (or rejects) complex SQL queries based on a privacy specification in extended SQL schema metadata. Our breakthrough is that, rather than executing $m$ times, we run just 1 query using \textbf{stochastic aggregation functions} that, in one go, compute $m$ partial results, one for each sub-sample. These results may be approximations, as the overall aggregate result will get noised anyway. The enabling insight is that each bit of a hashed primary key ({\bf key\_hash}) of the PU table can be seen to represent membership in a sub-sampled possible world. In Section 5, we show how a {\bf single SIMD instruction} can perform all additions in the below loop for a stochastic count aggregate (using $m$=64 here):

\vspace*{-0mm}\noindent{\small\begin{verbatim}
PacCountUpdate(uint64_t count[64], uint64_t key_hash)
  for(int j=0; j<64; j++) // for all 64 "possible worlds"
    count[j] += (key_hash >> i) & 1; // ++ if row exists in it
\end{verbatim}}\vspace*{-0mm}

The paper is organized as follows. In Section 2, we recap PAC privacy and PAC-DB and introduce SIMD-PAC-DB. Our DuckDB-PAC extension is outlined in Section 3, focusing on its SQL parser extension and algebraic rewriter. The semantics of PAC-DB and SIMD-PAC-DB are formalized in Section 4. We present SIMD-optimized stochastic aggregation algorithms in Section 5. The experimental evaluation, related work, and conclusion form Sections 6, 7, and 8.

\newpage
\section{From PAC-DB to SIMD-PAC-DB}\label{sec:2}

\vspace*{0mm}\noindent{\bf PAC Privacy.}
We first establish the background for PAC privacy~\cite{hanshen2023crypto}. Given a universe of records $U\subset\mathcal X^*$, PAC privacy first constructs $m$ subsets of $U$, denoted $\mathcal{S} = \{S_1,\ldots,S_m\}$, structured such that each record $x\in U$ belongs to $m/2$ subsets. Let $P_S$ denote the uniform distribution over $\mathcal{S}$. 
A mechanism $M:\mathcal X^*\to\mathcal Y$ satisfies $(\delta,\rho,P_S)$-PAC privacy under an attack criterion $\rho:\mathcal X^*\times\mathcal X^*\to\{0,1\}$ if an informed adversary $A$, who knows $P_S$ and $M$, cannot estimate $S$ from $M(S)$ with a success rate exceeding $1-\delta$ under criterion $\rho$. Formally, the adversary's posterior success rate is bounded as:
$$
1-\delta_{A}:=\mathrm{Pr}_{S\sim P_S,\, Y\gets M(S),\,\hat S\gets A(Y)}[\rho(\hat S,S)=1]\leq1-\delta.
$$
Even before observing $M(S)$, an informed adversary has a prior success rate, $1-\delta_0 :=\max_{Q}\Pr_{S\sim P_S,\hat S\sim Q}[\rho(\hat S,S)=1]$, determined by $P_S$ and $\rho$. Because this prior is fixed, we can restrict the adversary's posterior success rate $(1-\delta_A)$ by limiting their \emph{posterior advantage}, the divergence between their prior and posterior success rates. Using KL divergence, this advantage is bounded by the mutual information (MI) between the input and output~\cite{hanshen2023crypto}:
\begin{equation}\label{eq:delta<=mi}
D_\text{KL}(\bm{1}_{1-\delta_A}\,\|\,\bm{1}_{1-\delta_0})
\leq \mathrm{MI}(S;M(S)).
\end{equation}
Hence, one can privatize an algorithm against arbitrary adversaries by setting a MI budget $B>0$ and enforcing $\mathrm{MI}(S;M(S))\leq B$.

For a real-valued function $f:\mathcal X^*\to\mathbb R$, we can achieve this by calibrating noise to the function's variance under $P_S$. Given a budget $B$, we release $\tilde f_B(S)=f(S)+\Delta$, where
$\Delta\sim\mathcal N\left(0, \frac{1}{2B}\mathsf{Var}_{S\sim P_S}[f(S)]\right)$.
This ensures $\mathrm{MI}(S;\tilde f_B(S))\leq B$, as shown by \citet{sridhar2025pac}.

For a sequence of possibly adaptive functions $f_1,\ldots,f_d:\mathcal X^*\to \mathbb R$, we aim to bound the MI between $S\sim P_S$ and $(\tilde f_{1,B}(S),\ldots,\tilde f_{d,B}(S))$ to limit adversarial inference after observing $d$ releases evaluated on a persistent secret $S$. This requires adaptive noise calibration~\cite{zhu2026pacadversarial}, where we maintain a posterior distribution $P_i\in\mathbb R^m$ of the secret, initialized as $P_0=P_S$. For the $i$-th release, the mechanism evaluates the variance of the $m$ possible outputs of $f_i$ under $P_{i-1}$ (rather than $P_S$) 
and release $\tilde f_{i,B}=f_i(S)+\Delta_i$ where $\Delta_i\sim\mathcal N(0,\frac{1}{2B}\mathsf{Var}_{S\sim P_{i-1}}[f_i(S)])$
for a per-release budget $B$. After the noisy release, the mechanism performs a Bayesian update to compute the posterior distribution $P_{i}$ for the next release. This ensures the total MI over $d$ adaptive releases compose linearly~\cite{zhu2026pacadversarial}: $\mathrm{MI}(S; (\tilde f_{1,B}(S),\ldots,\tilde f_{d,B}(S)))\leq dB$.

For any inference attack, the total MI guarantee directly translates into a concrete upper bound on the adversary's success rate via Eq.~\eqref{eq:delta<=mi}. By first measuring the adversary's prior success rate ($1-\delta_0$) and bounding their posterior advantage by the total MI, we can provably constrain their final inference capabilities. 
Consider a standard membership inference attack (MIA)~\cite{carlini2022membership} where the adversary attempts to infer whether record $x^*\in U$ is present in the secret subset $S$. Because each record appears in exactly half of the subsets, the adversary's prior success rate before observing any releases is exactly 50\% for any target $x^*$, \emph{regardless of $m$}. Applying Eq.~\eqref{eq:delta<=mi}, we can rigorously cap the MIA success rate: a total MI budget of $1/4$ limits MIA success to $\approx 84\%$; tightening MI to $1/128$ lowers the bound to 53\%. Crucially, this guarantee applies simultaneously to all records in $U$ because of a uniform prior across the universe. Finally, we note that as this bound assumes a computationally unbounded adversary with full knowledge of $P_S$, it provides a conservative privacy guarantee that naturally extends to realistic attackers.
\newpage

\vspace*{0mm}\noindent{\bf Relational Privacy Unit.}
To enable PAC privacy for relational databases, we first designate a table as the \emph{privacy unit} (PU) table to represent the universe $U$. Every single row in this table acts as a PU, which is the atomic unit of privacy protection --- intuitively, it represents the smallest entity whose presence or absence in the database must be protected. 
Consequently, the $m$ subsets $\mathcal{S} = \{S_1, \ldots, S_m\}$ are constructed by sub-sampling these PUs. Because relational databases rely on referential integrity, evaluating a query over a sub-sample of PUs requires more than just sub-sampling the PU table; it necessitates the cascading inclusion of rows in other tables that are transitively linked to the sampled PUs via foreign key relationships. We therefore require a mechanism to associate arbitrary rows with the PUs that contribute to them.

Let $T$ be a relation scanned by a query. If $T = U$, each tuple trivially corresponds to exactly one PU. Otherwise, $T$ may be linked to $U$ through a chain of foreign key relationships:
\[
T \xrightarrow{fk_1} T_1 \xrightarrow{fk_2} \cdots \xrightarrow{fk_k} U.
\]
In this case, each tuple $t \in T$ is associated with the unique PU identifier $u.pk$ reached by following the foreign key path. 

\vspace*{2mm}\noindent{\bf PAC-DB Example.}
In this TPC-H example, \texttt{customer} is designated as the PU and PAC-DB generates $m$ independent random samples (``worlds") of the PU population (this paper uses $m=64$), where each PU is independently included with probability $1/2$ in each sample. Then, it runs the query against each sample. We also consider the purchasing behavior of customers sensitive: if \texttt{orders} or \texttt{lineitem} are queried, we need to perform equi-joins on their FK-PKs towards the sub-sampled \texttt{customer} PU table, and then execute the original query logic. All this results in $m$ query executions. From the $m$ outputs, PAC-DB computes the variance and adds calibrated Gaussian noise to produce the released result. While PAC-DB establishes a foundation for PAC privacy, its proof-of-concept implementation$^1$ prioritized semantic clarity over execution efficiency. 
Nevertheless, the initial evaluation~\cite{sridhar2026pacdb} shows that PAC-style noise yields significantly higher accuracy than differential privacy baselines for the same level of protection.

\begin{customsql}[caption={PAC-DB rewrite of TPC-H Q01.}, label={lst:pac-db}]
CREATE TEMPORARY TABLE samples AS (
    SELECT  s.sample_id, customer.rowid AS row_id,
        (RANDOM() > 0.5) AS random_binary
    FROM (SELECT range AS sample_id FROM range(128)) s
    JOIN customer ON TRUE);
    
FOR $sample=0 TO 128:
    SELECT l_returnflag, l_linestatus,
        sum(l_quantity) AS sum_qty,  
        avg(l_quantity) AS avg_qty, ...
        count(*) AS n_orders
    FROM lineitem JOIN orders ON l_orderkey = o_orderkey
                  JOIN customer ON o_custkey = c_custkey
                  JOIN samples ON row_id = customer.rowid
    WHERE l_shipdate <= CAST('1998-09-02' AS date)
    AND s.random_binary = TRUE AND s.sample_id = $sample
    GROUP BY l_returnflag, l_linestatus
# pick a sample s to return & noise its aggregate results  
# pac_noise() uses all 128 outcomes to get variance 
\end{customsql}
\newpage

\begin{customsql}[caption={SQL for TPC-H Q01, as rewritten by SIMD-PAC-DB.}, label={lst:optimized}]
SELECT l_returnflag, l_linestatus,
  pac_noised_sum(pac_hash(hash(o_custkey)), l_quantity),
  pac_noised_avg(pac_hash(hash(o_custkey)), l_quantity),
  ..., pac_noised_count(pac_hash(hash(o_custkey))) 
FROM lineitem JOIN orders ON l_orderkey = o_orderkey 
WHERE l_shipdate <= CAST('1998-09-02' AS date)
GROUP BY l_returnflag, l_linestatus
\end{customsql}

\vspace*{1mm}\noindent{\bf SIMD-PAC-DB. } 
This paper pursues a completely different way of implementing PAC-DB. It only executes a single query, making use of new {\bf stochastic aggregate functions} that we introduce, such as, \texttt{pac\_noised\_count(UBIGINT):UBIGINT}. It performs the work of computing an array \texttt{count[64]} of $m$=64 intermediate aggregates, where aggregate $j$ is updated iff the row belongs to a PU that was inside the 50\% sub-sample of the $j$-th ``possible world".
Its C++ pseudo-code implementation is shown in the \texttt{PacCountUpdate()} code on page 1.
Membership of world $j$ is encoded by bit $j$ of the hashed PK of the PU row \texttt{hash(c\_custkey)}. 
When \texttt{pac\_noised\_count()} finalizes, it uses the 64 aggregates to compute PAC-DB noise, and returns a single, noised, value. 

In SIMD-PAC-DB, we use $m=64$, as this aligns with DuckDB's \texttt{uint64\_t} hash numbers.\footnote{The reduction mostly affects utility rather than membership privacy guarantees. Extending to $m$=128 is possible, and worst-case would double the cost of our stochastic-aggregates --- but is likely less due to increased SIMD efficiencies of scale.}
A good hash function produces  integers with on average 50\% 0- and 50\% 1-hash bits, binomially distributed ($n$=64, $p$=0.5){\color{blue}}. We wrap DuckDB's hash function in our own \texttt{pac\_hash(UBIGINT):UBIGINT} function for two reasons: (i) this allows us to change the hash function on every query, using a rehash which virtually shuffles the complete universe of 64 possible worlds, allowing to consider a per-query privacy budget, rather than a session budget, and (ii) we can cheaply transform the binomially distributed hash number in a random hash that is {\em guaranteed} to have 32 0- and 32 1-bits; 
this stabilizes our stochastic aggregates and guarantees a uniform prior success rate of 50\% under MIA.

The \texttt{pac\_noised\_count()} function is the fused version of two functions: (i)
 \texttt{pac\_count(UBIGINT):list$<$UBIGINT$>$} that returns the 64 counters as a list\footnote{We could also have used fixed-size arrays: {\tt UBIGINT[64]}, but use lists as DuckDB provides lambda functions for these, which we use to evaluate arbitrary expressions.}, and (ii) 
\texttt{pac\_noised(list$<\!\!T\!\!>$):$T$} which measures the variance in the 64 inputs, takes the value from the possible world $j^*$ and returns \texttt{count[$j^*$]*2 + {\em noise}} (the \texttt{*2} is because each counter is expected to only see half of the tuples).
If multiple PAC aggregates are computed by a query, the query always use the same, randomly selected, possible world $j^*$ for better utility and query-level privacy.

\vspace*{2mm}\noindent{\bf PU-key Joins.} The original Q01 only scans the \texttt{lineitem} table, but in order to know to which possible world each row belongs, we need to fetch the PK  \texttt{c\_custkey} of the PU table \texttt{customer}. 
For this purpose, our rewriter adds \texttt{JOIN orders ON l\_orderkey = o\_orderkey}.
We note that \texttt{orders} contains the FK \texttt{o\_custkey}, which is all we need, hence a join to \texttt{customer} is not required.
While SIMD-PAC-DB executes just one query and is much faster than PAC-DB, Figure 1 shows it is slower than default DuckDB because (i) the cost of these PU-key joins, (ii) PAC aggregates are more costly than default ones. In Section~\ref{sec:simd}, we describe optimizations targeting SIMD execution that limit their overhead.

\section{DuckDB-PAC Extension}
\label{sec:implementation}
We implement SIMD-PAC-DB as the \texttt{pac} community extension, which can be installed in the DuckDB~\cite{raasveldt2019duckdb,raasveldt2020embedded} shell: \\ \texttt{install pac from community; load pac;}\footnote{This will be realized by March 15 after the DuckDB v1.5 release - DuckDB-PAC can always be manually compiled from \href{https://github.com/cwida/pac}{github.com/cwida/pac}}.

\subsection{PAC Parser}
SIMD-PAC-DB introduces an extended SQL syntax for declarative statements that define table relationships and designate columns that require privacy protection. Our DuckDB-PAC extension implements this with a parser extension that validates metadata and stores it separately from the database schema.

\vspace{2mm}\noindent{\bf PAC keys} serve as identifiers for PUs. Since PK constraints are computationally expensive to maintain and are often absent in analytical systems, PAC keys provide a lightweight alternative. 
Queries that scan a PU table or a dependent table must perform an aggregation where PAC noise is added --- queries that return unaggregated protected columns are rejected. As protected columns cannot be directly returned, they cannot be used as \texttt{GROUP BY} keys --- unless the result is re-aggregated on non-protected columns.
By default, all columns of a PU table are considered protected, but a specification can override this with the listed columns only.

\begin{simplecode}[language=SQL, caption={SQL syntax for PU definition, PAC keys and links. Also supports ALTER/DROP, see github.com/cwida/pac}, label={lst:pac-protected}]
CREATE PU TABLE table_name (
    column_name data_type [, ...],
    PAC_KEY (key_column [, key_column]*),
    PROTECTED (column_name [, column_name]*));
    
CREATE TABLE table_name (
    column_name data_type [, ...],
    PAC_LINK (local_column [, local_column]*)
        REFERENCES referenced_table 
        (referenced_column [, referenced_column]*));    
\end{simplecode}

\vspace*{0mm}\noindent{\bf PAC Links.}
PAC links define FK relationships between tables to propagate privacy. Unlike database-enforced FKs, PAC links are metadata-only and do not incur constraint checking overhead. Multi-column PAC links are supported. PAC links must not form cycles. All local and referenced columns named in a \texttt{PAC\_LINK} declaration are considered protected. 

\vspace*{2mm}\noindent{\bf Query Validation.}
When a query is issued, DuckDB-PAC traverses its plan to determine whether it should undergo a PAC transformation. There are 3 cases we identify:

\textit{Inconspicuous queries}: do not reference a PU (or a PAC linked table): these are executed in unmodified form.

\textit{Rejected queries}: if a query is directly returning a protected column without an aggregate on top. Also, when multiple protected tables are used in a query, they {\bf must} be equi-joined {\bf only} over exact \texttt{PAC\_LINK}s, otherwise they are rejected. Queries accessing PU data containing operators which are not yet supported --- such as window functions, recursion, or \texttt{NOT EXISTS} (negation) other than over \texttt{PAC\_LINK}s, also get rejected. 

\textit{Rewritable queries}: otherwise, if a query references a PU (either directly or indirectly), it will be rewritten to use PAC aggregates.

\subsection{PAC Rewriter}
The PAC rewriter --- corresponding to DuckDB-PAC's implementation of Algorithm~\ref{alg:pac-rewrite} ---  transforms a standard query plan into its SIMD-PAC-DB privatized equivalent, using the metadata defined in Listing~\ref{lst:pac-protected}. To enable global optimization after our rewrites, we perform all the plan modifications in DuckDB's pre-optimization phase. 
The rewrite proceeds in two phases: the \textit{top-down} and \textit{bottom-up}.

\vspace*{2mm}\noindent{\bf Top-down.} When the traversal encounters a table $T$ PAC linked to the PU table $U$, we insert the necessary chains as described in Section~\ref{sec:2}, computing hash column \textsf{pu} as \texttt{pac\_hash(H(pk))} (lines 3--5 of Algorithm~\ref{alg:pac-rewrite}). This is also pursued with materialized CTEs, which, in DuckDB's algebra, can be referenced in multiple places. If an Aggregate is enclosed in a \emph{correlated} subquery, we recurse into it and also add PU joins there; unless the correlating condition is a PAC link, in which case the outside \textsf{pu} is already visible.

\vspace*{2mm}\noindent{\bf Bottom-up.} After recursing on all children (lines 11--13 of Algorithm~\ref{alg:pac-rewrite}), the bottom-up phase visits each node and projects the \textsf{pu} column. Then, we rewrite standard aggregates $\mathsf{pac\_\!\!<\!\!aggr\!\!>\!\!(pu,..)}$, if their input rows are sensitive but their aggregation key is not (line 14). Note that e.\:g., the subquery of TPC-H Q13 groups by \texttt{c\_custkey}. We do not rewrite there; the query is made privacy-safe by the presence of a PAC aggregate at the top of the plan.

\vspace*{2mm}\noindent{\bf Categorical Rewrites.} Expressions involving (PAC) aggregates may appear in SELECT, GROUP BY, or ORDER BY lists, which in the logical plan will surface in the Project operator ($\pi$ here is non-duplicate eliminating, used for selecting columns and expression calculation), but may also appear in WHERE conditions. In the latter case, there must be a scalar subquery, which may or may not be correlated with the outer query.
Rather than immediately noising each aggregate with e.\:g. \texttt{pac\_noised\_count()}, our rewriter (lines 14-16), uses unfused PAC aggregates such as \texttt{pac\_count()}, which return 64 yet unnoised values.
We can iterate over these values, evaluating the expression in which the aggregate is embedded, using a lambda provided by DuckDB's \texttt{list\_transform()} (lines 17--19). If the expression involves multiple aggregates, they are first combined into a list of structs with the \texttt{list\_zip()} SQL function. 

For projections, the \texttt{pac\_noised()} function is put on top of the \texttt{list\_transform()} to produce a single noised result value (line 21). TPC-H Q08, 14 perform such aggregate expression calculations.

Noising boolean choices made in a WHERE clause differs from noising an aggregate: filters need to make a yes/no decision whether to return a row -- a binary rather than a continuous outcome.
We implement "probabilistic filtering" in the \texttt{pac\_filter(list$\!<\!$bool$\!>\!$ b):bool} function (used in line 26). It measures the \#true fraction in the 64 counters, and uses it as probability for making a random (i.\:e. noised) choice between true and false. Note that this reveals no information regarding the choice of the underlying secret world. We leave filter-after-aggregation out of the formal analysis in the next section, because here SIMD-PAC-DB  deviates from PAC-DB~\cite{sridhar2026pacdb}. There, the union of returned rows in the 128 query runs was sampled in a way that could select the same row multiple times -- we think avoiding duplicates helps usability.

However, when there is another PAC aggregate higher up in the plan (e.\:g. TPC-H Q17), an eager yes/no choice is not needed: rather, we AND \textsf{pu} bit-by-bit with the 64 booleans (line 24). The renamed \textsf{pu} can then be used by the upstream PAC aggregate functions. As an optimization, we filter out rows where this updated \textsf{pu}=0, since such rows do not play a role in any possible world.
For this AND between a hash and a \texttt{list$<\!$boolean$\!>$} we introduced the function \texttt{pac\_select()}. Since lambda expressions that iterate over lists introduce some overhead we also introduced for simple comparison expressions $<,\le,=,\neq,>,\ge$ fused versions like \texttt{pac\_select\_gt(hash, col, list$<\!\!T\!\!>$):bool} that perform the comparison ("\texttt{\_gt}": \texttt{col} $>$ \texttt{aggr}) followed by an AND with the PU hash in one go.
Note that the \texttt{pac\_select()} case falls inside the query fragment we formally analyze in the next section, where SIMD-PAC-DB upholds PAC-DB semantics.

\vspace*{2mm}\noindent{\bf NULLs.} During aggregation, each PAC aggregate maintains a 64-bit \textit{accumulator}: for each contributing tuple with hash $h$, we OR this accumulator ($K \gets K \mathbin{|} h$). At finalization, the accumulator encodes which of the subsamples received at least one contribution. Then, we employ a probabilistic mechanism returning NULL with probability $P(\mathsf{NULL}) = \frac{64 - \textbf{popcount(k)}}{64}$, independent of the secret realization. This naturally scales with data sparsity. When the aggregate does not return NULL, counters corresponding to unset bits are treated as zeros for noise calibration during PAC privatization.

\vspace*{0mm}
\begin{algorithm}[H]
    \small
    \caption{The \textsf{PacRewrite} algorithm. Top-down: insert PU joins and derive hash expressions, also visiting CTEs.
    Bottom-up: replace aggregates and rewrite categorical expressions, i.e., expressions involving PAC aggregates produced by (correlated) subqueries.}
    \label{alg:pac-rewrite}
    \begin{algorithmic}[1]
    \Function{PacRewrite}{node $V$, PU metadata $\mathcal{M}$}
      \Comment{\textit{--- Top-down phase}}
      \If{$V = \mathrm{Scan}(T)$ \textbf{and} $\mathrm{FKPath}(T,U)\in\mathcal{M}$ \label{line:linkpustart}
          \textbf{and} no join to $T_K$ yet}
        \State $T \xrightarrow{fk_1} T_1 \cdots
               T_k \xrightarrow{fk_k} U \gets \mathrm{FKPath}(T,U)$
               \Comment{add PU key joins}
        \State $V \gets V \bowtie_{fk_1}
               \!\pi_{\mathsf{pk}_1,fk_2}(\mathrm{Scan}(T_1))
               \!\bowtie_{fk_2}\!\cdots\!\pi_{fk_k}(\mathrm{Scan}(T_k))$
        \State $V \gets \pi_{\mathsf{cols}(V),\mathsf{pac\_hash}(fk_k)\rightarrow\,\mathsf{pu}}(V)$
               \Comment{compute PU hash}
      \EndIf
      \If{$V = \mathrm{CTE}[\mathrm{body}]$ \textbf{and} body refs.\ PU tables}
        \State $\mathrm{body} \gets
               \Call{PacRewrite}{\mathrm{body},\mathcal{M}}$
        \State $V \gets \mathrm{CTE}[
               \pi_{\mathrm{cols}(\mathrm{body}),\mathsf{pu}}(\mathrm{body})]$
               \Comment{CTE hash propagation}
      \EndIf
      \Comment{\textit{--- Recurse on children}}
      \For{each child $V_i$ of $V$}
        \State $V_i \gets \Call{PacRewrite}{V_i,\mathcal{M}}$
      \EndFor \label{line:linkpuend}
      \Comment{\textit{--- Bottom-up phase}}
      \If{$V\!=\!\gamma_{G,\mathsf{agg}_1(\mathsf{e}_1),..,\mathsf{agg}_d(\mathsf{e}_d)}(T)\!$ \textbf{and} $\exists T.\mathsf{pu}$ \textbf{and} $\neg\mathsf{sensitive}(G)\!$}\label{line:replaceaggregationstart}
        \State $V \gets \gamma_{G,\,\mathsf{pac\_agg}_1(\mathsf{pu},\mathsf{e}_1),..,\,\mathsf{pac\_agg}_d(\mathsf{pu},\mathsf{e}_d)}(T)$
      \EndIf \label{line:replaceaggregationend}
      \If{$V \in\{\sigma_\phi, \pi_{\phi_*} 
      \}$ \textbf{and}
          $\exists \phi_i(s_i), 1\!\leq\! i\!\leq\!m$, using $\mathsf{pac\_aggr}'s\ s_i$}
        \State $\mathbf{e}_i^1,\mathbf{e}_i^2,\ldots \gets s_i$
              \Comment PAC aggregate exprs in set $s^i$
        \State $\mathbf{r_i} \!\gets\!
               \Lambda_{(..\mathbf{e}_i^*..)}[\phi_i]$ 
               \Comment list\_transform(list\_zip($..\mathbf{e}_i^*..$),$x\!:\!\phi_i(x)$)
        \If{$V=\pi$} \label{line:addnoisestart}
                 \Comment add noise to aggr expr
                   \State $V \!\!\gets\!\!\pi_{*/\phi_1..\phi_m,\mathsf{pac\_noised}(\mathbf{r_1}) \rightarrow\,\phi_1,.,\mathsf{pac\_noised}(\mathbf{r_m}) \rightarrow\,\phi_m}\!\!(
                 V.\mathrm{child})$ \label{line:addnoiseend}
        \ElsIf{$V=\sigma$}
          \If{outer pac\_agg.\ exists}          \label{line:selectioncorrelatedsubquerystart}
            \Comment restrict $\mathsf{pu}$ for outer aggr
            \State $V \gets \sigma_{\mathsf{pu}\ne 0}(
                 \pi_{*/\mathsf{pu},\mathsf{pac\_select}(\mathsf{pu},\mathbf{r_1})\,\rightarrow\,\mathsf{pu}}\!(V.\mathrm{child}))$ \label{line:selectioncorrelatedsubqueryend}
           \Else
                 \Comment noised row filtering
          \State $V \gets \sigma_{\mathsf{pac\_filter}(\mathsf{pu},\mathbf{r_1})}(
                 V.\mathrm{child})$
          \EndIf
      \EndIf
   \EndIf
   \State \Return $V$
\EndFunction
    \end{algorithmic}
  \end{algorithm}
\newpage

  \vfill\newpage

\section{SIMD-PAC-DB Semantics}
In this section, we formalize the query semantics of SIMD-PAC-DB. 
We first present a relational-algebra formalization of a baseline execution model for PAC-DB~\cite{sridhar2026pacdb} that materializes the $m$ possible worlds, and state formal privacy guarantees.
We then introduce the semantics of SIMD-PAC-DB and highlight its advantages over the PAC-DB formulation. 
Crucially, we prove that these semantics provide the same privacy guarantees as PAC-DB, computing the exact same variances and noise without the overhead of realizing $m$ physical worlds. 

\vspace*{2mm}
\noindent
\textbf{Query Class.}
Let $\mathbb{I}$ be the set of base relations. 
Then, query class
$$
\mathcal{Q}
\;:=\;
\pi_{G;\; \mathsf{e}_1 \rightarrow \alpha_1,\; \dots,\; \mathsf{e}_c \rightarrow \alpha_c}
\!\left(
    \gamma_{G;\mathsf{Aggs}}(\mathcal{T}(\mathbb{I}))
\right).
$$
Here, $G$ is a (possibly empty) set of grouping attributes, $\mathsf{Aggs}$ is a collection of aggregation functions, and 
$\mathsf{e}_i \rightarrow \alpha_i$ denotes a projection expression with alias $\alpha_i$.
The subquery $\mathcal{T}$ is defined inductively. Let $T^1_s$ and $T^2_s$ denote base or intermediate relations that admit a foreign-key chain to the PU table. Then $\mathcal{T}$ can be constructed by arbitrary compositions of the following sub-expressions:

\begin{enumerate}[leftmargin=*,label=(\alph*)]
    \item $E = \pi_{\mathsf{cols}(T^2_s)}(\sigma_{\mathsf{p}}\!\left(T^2_s \;\mathsf{Join}_{G_1}\;\gamma_{G_1;\mathsf{Aggs}_1}(T^1_s))
\right)$, where $\mathsf{Join}_{G_1}$ denotes a join on the attribute set $G_1$, and $\mathsf{p}$ is a predicate over inner aggregation attributes and attributes of $T^2_s$. 
Specifically, it may represent (i) an inner join ($\Join_{\theta}$), 
(ii) a left outer join ($\leftouterjoin_{\theta}$), or 
(iii) a Cartesian product when $G_1 = \emptyset$.
    \item $F(\mathbb{T})$, where $F$ denotes an arbitrary relational algebra expression over a collection of base or intermediate relations $\mathbb{T}$, provided that it does not merge or combine tuples originating from different PU identifiers.
    \item The relations $T^1_s$ and $T^2_s$ are themselves produced by subqueries of type $F$, and their outputs may in turn serve as inputs to further expressions of type $F$. Therefore, the query class supports multiple inner aggregations. 
\end{enumerate}

This requires that every join in $\mathcal{Q}$ between two sensitive base or intermediate relations exactly satisfies the PAC Link constraints. 

\subsection{PAC-Private Databases} \label{sec:packdbsemantics}
\sloppy
We introduce two  auxiliary functions and expressions.

\vspace*{2mm}
\noindent
\textbf{SamplePU$(U)$:}
This function outputs $m$ sampled instances of the PU table $U$, denoted by $\{S_j\}_{j=1}^m$. 

\vspace*{2mm}
\noindent
\textbf{Noise function} $\mathsf{pac\_noised}(\mathsf{col}, j^*, B) \rightarrow \mathbb{R}$:
This stateful function is used within the projection operator $\pi$. It takes as input a column $\mathsf{col}$ that takes a $m$-dimensional vector in $\mathbb{R}^m$, the secret world index $j^*$, and a per-cell MI budget $B$. Let $y_j := \mathsf{col}[j]$ for $1 \leq j \leq m$, it:
\begin{enumerate}[leftmargin=*]
    \item computes the variance of the outputs under $P$, the current posterior distribution of the secret: $s^2 \leftarrow \mathsf{Var}_{j\sim P}[y_j]$. 
    \item sets the calibrated noise variance $\Delta \leftarrow s^2/(2B)$.
    \item outputs $\widetilde{y} \leftarrow y_{j^*} + \mathcal{N}(0, \Delta)$.
    \item performs a Bayesian update on the global vector $P$ given the observation $\widetilde{y}$:
    $P \leftarrow (P \odot W)/\mathsf{sum}(P\odot W)$,
    where $\odot$ denotes element-wise multiplication and $W \in \mathbb{R}^m$ is the likelihood vector with elements $W_j = \exp(-(\widetilde{y} - y_j)^2/(2\Delta))$.
\end{enumerate}

Given a query $Q \in \mathcal{Q}$, a PAC-DB engine produces the PAC-private query result  through the following compilation procedure:

\noindent \textbf{(1) Generate $m$ sampled database instances.} For each relation $I \in \mathbb{I}$ that has a foreign-key path to the PU table $U$, it joins $I$ with $U$ along the sequence of equi-joins induced by the corresponding $fk = pk$ relationships. It then projects the attributes $\mathsf{col}(I)$ together with the privacy-unit attribute.
Let $\{S_j\}_{j=1}^{m} := \mathsf{SamplePU}(U)$. 
For each relation $I$ and each sample $j$, we define
\[
I^{(j)} :=
\begin{cases}
I \Join_{S_j.pk} S_j, & \text{if $I$ has a foreign-key path to $U$,} \\[0.5em]
I, & \text{otherwise.}
\end{cases}
\]
The $j$-th sampled database instance is then defined as:
\[
\mathbb{I}^{(j)}
:=
\{\, I^{(j)} \mid I \in \mathbb{I} \,\}.
\]

\noindent
\textbf{(2) Execute $m$ queries and add noise.}
Evaluate $Q$ on each sampled instance and take the multiset union:
\[
T_1 
:= 
\uplus_{j=1}^{m}
Q(\mathbb{I}^{(j)}).
\]
Then, aggregate the $m$ values of the same group key as a $m$-dimensional vector. This is achieved via a $\mathsf{List}(\cdot)$ aggregation function, which collects the input values into an ordered list: 
\[
T_2
:=
\gamma_{G; \mathsf{List}(\alpha_1) \rightarrow \alpha'_1, \cdots, \mathsf{List}(\alpha_c) \rightarrow \alpha'_c}(T_1).
\]
Finally, sample a secret world index uniformly at random $j^* \stackrel{\$}{\leftarrow} \{1, \ldots, m\}$, initialize a global probability vector $P \in \mathbb{R}^m$, representing the posterior distribution over the subsets after observing noisy releases, to a uniform distribution $P = (1/m, \ldots, 1/m)$. Subsequently, apply the noise mechanism to each aggregated attribute:
\[
\mathsf{Output}_{\mathrm{PAC-DB}}(\mathbb{I}, Q)
:=
\pi_{G,\;
\mathsf{pac\_noised}(\alpha'_1, j^*, B), \ldots,
\mathsf{pac\_noised}(\alpha'_c, j^*, B)
}(T_2).
\]

\vspace*{1mm}
\noindent
\textbf{Privacy Guarantee.} By adaptively calibrating the injected noise to the variance across the $m$ sampled instances under the posterior distribution of the secret, this compilation procedure rigorously bounds the adversarial advantage of an arbitrary inference attack.

\begin{theorem}
\label{thm:post_adv}
For any inference attack and an informed adversary $A$, let $1-\delta_0$ denote the prior success rate and $1-\delta_A$ denote the posterior success rate after observing $\mathsf{Output}_{\mathrm{PAC-DB}}(\mathbb{I}, Q)$. The adversary's posterior advantage is strictly bounded by:
$$D_{\mathrm{KL}}(\bm{1}_{1-\delta_A} \,\|\, \bm{1}_{1-\delta_0}) \leq \mathrm{MI}(j^*;\mathsf{Output}_{\mathrm{PAC-DB}}(\mathbb{I}, Q)) \leq cB.$$
\end{theorem}

\vspace*{1mm}
Theorem~\ref{thm:post_adv} provides a systematic method to derive concrete privacy guarantees against any arbitrary inference attack. Given an inference attack, we first evaluate the prior success rate ($1-\delta_0$) of an informed adversary who knows the distribution from which the secret subset $S_{j^*}$ is sampled from, that is, knows $(S_1, \ldots, S_m)$ generated by $\mathsf{SamplePU}(U)$. Next, we plug this prior $1-\delta_0$ and the total MI bound $cB$ into Theorem~\ref{thm:post_adv} to solve for a strict upper bound on the posterior success rate ($1-\delta_A$).

As a concrete example, consider a \emph{query-level MIA} where an adversary attempts to infer if a PU is \emph{ever} used to evaluate the $c$ cells to produce $\mathsf{Output}_{\mathrm{PAC-DB}}(\mathbb{I}, Q)$.
When $\mathsf{SamplePU}$ is instantiated with a balanced subset construction (i.e., $\forall x\in U$, $|\{i: x \in S_i\}| = m/2$), any PU appears in exactly half of the possible subsets.
Since a persistent secret subsample $j^*$ is used to produce all the cell outputs, an informed adversary's prior success rate is exactly $50\%$. By substituting $1-\delta_0 = 0.5$ and the accumulated MI $cB$ into the theorem, we can numerically solve for the range of $\delta_A$, directly yielding the concrete limits on MIA success established in Section~\ref{sec:2}.

\subsection{SIMD-PAC-DB Compiler}
SIMD-PAC-DB replaces the costly sampling procedure of PAC-DB with a lightweight hash-based sampling mechanism that avoids joins with a sample table. Second, we observe that sampling need not be applied at the level of base relations. Instead, the $m$ logical samples can be introduced
lazily, only at operators that may induce shared ownership—e.g., aggregations. Third, and most importantly, rather than executing $m$ separate query evaluations, SIMD-PAC-DB evaluates a single query by incorporating novel aggregation functions that implicitly capture the effect of the $m$ samples. 

\vspace*{2mm}\noindent\textbf{Vector-lifted Expression, and Lambda Functions.} Let $\mathsf{e}(\mathsf{cols})$ denote an arbitrary scalar expression defined over a set of columns $\mathsf{cols} = \{\mathsf{col}_1, \ldots, \mathsf{col}_a\}$, where each column $\mathsf{col}_i$ takes a scalar value, and the expression outputs a scalar in $\mathbb{R}$. We define the vector-lifted version of $\mathsf{e}$, denoted by 
$\vec{\mathsf{e}}(\mathsf{cols})$, as the pointwise extension of $\mathsf{e}$ to mixed scalar/vector inputs. Specifically, assume that all vector-valued columns are $m$-dimensional, $\vec{\mathsf{e}}(\mathsf{cols})$ produces an $m$-dimensional vector:
$$\vec{\mathsf{e}}(\mathsf{cols})
=
\big(
\mathsf{e}(\{\mathsf{col}_i[1]\}_{i=1}^a),
\ldots,
\mathsf{e}(\{\mathsf{col}_i[m]\}_{i=1}^a)
\big).$$ 
For scalar-valued columns, we adopt implicit broadcasting:
if $\mathsf{col}_i$ is scalar, then $\mathsf{col}_i[j] = \mathsf{col}_i$ for all $1 \leq j\leq m$.

We emphasize that vector-lifted expressions are evaluated efficiently by leveraging higher-order SQL list operators—namely $\mathsf{list\_transform}$ that applies a lambda, and $\mathsf{list\_zip}$. Formally, the lifted expression $\vec{\mathsf{e}}$ can be expressed as:
\begin{equation} \label{vectorexpressionlambda}
\begin{split}
\vec{\mathsf{e}}(\mathsf{cols})
=
\mathsf{list\_transform}
\Big(
    \mathsf{list\_zip}(\mathsf{col}_1, \ldots, \mathsf{col}_a), \\
  \lambda x.\,\mathsf{e}\big(x[1], \cdots, x[a]\big)
\Big).
\end{split}
\end{equation}
Here, $\mathsf{list\_zip}$ aligns the input columns element-wise into tuples, and $\mathsf{list\_transform}$ applies the scalar expression $\mathsf{e}$ to each tuple via the lambda function, thereby producing the vectorized result.

\vspace*{2mm}
\noindent \textbf{PAC Hash $\mathsf{pac\_hash}_K(\mathsf{col}) \rightarrow \{0,1\}^m$:}
We define $\mathsf{pac\_hash}$ as a keyed hash function that assigns an $m$-bit vector to each PU, where $K$ is a randomly sampled key drawn independently for each query. Crucially, $\mathsf{pac\_hash}_K$ is designed to be the vector-functional equivalent of the subset construction strategy used by $\mathsf{SamplePU}$.

For a given PU table $U$, defining the $i$-th implicit logical subset as $S_i = \sigma_{\mathsf{pac\_hash}_K(pk)[i]=1}(U)$ ensures that the joint distribution of the resulting subsets $(S_1, \ldots, S_m)$ over the random choice of $K$ is strictly identical to the distribution of subsets produced by $\mathsf{SamplePU}(U)$. $\mathsf{pac\_hash}()$ further applies a query-specific re-hash, such that for every query, the worlds are fully re-created -- something that would be  expensive in PAC-DB's sample table approach.


\vspace*{2mm}
\noindent \textbf{PAC Aggregation $\mathsf{pac\_agg}(\mathsf{pu}, \mathsf{e}(\mathsf{cols})) \rightarrow \mathbb{R}^m$.} It takes as input the hash column $\mathsf{pu}$ and a scalar expression $\mathsf{e}$. Given an input relation $T$, it produces an $m$-dimensional vector $\{r_1, \cdots, r_m\}$ defined by for all $1 \leq j\leq m$, where $\mathsf{agg} \in
\{\mathsf{sum}, \mathsf{count}, \mathsf{avg}, \mathsf{min}, \mathsf{max}\}.$
\[
r_j
=
\mathsf{agg}
\Big(
\{
\mathsf{e}(t[\mathsf{cols}])
\mid
t \in T,\;
\mathsf{pu}[j] = 1
\}
\Big),
\]

\vspace*{2mm}\noindent \textbf{PAC Selection $\mathsf{pac\_select}(\mathsf{pu},\vec{\mathsf{p}}(\mathsf{cols})) \rightarrow \{0,1\}^m$}
takes as input a hash column $\mathsf{pu} \in \{0, 1\}^m$ and a Boolean-vector expression
$\vec{\mathsf{p}}(\mathsf{cols}) \in \{0,1\}^m$, 
and returns their bitwise AND.  
\newpage

\vspace*{0mm}
\noindent{\textbf{Compiler.}}
Given a query $Q \in \mathcal{Q}$, 
the SIMD-PAC-DB engine produces a PAC-private result via the 
following compilation procedure.

\noindent\textbf{(1) PU propagation ($\mathsf{Re}^{\mathsf{SIMD}}_1(Q)$).}
As shown in Lines~\ref{line:linkpustart}--\ref{line:linkpuend} of 
Algorithm~\ref{alg:pac-rewrite}, we traverse the 
expression of $Q$ in a top-down manner. For each operator, if it's input relation $T$ has foreign keys participating in a path contained in $\mathrm{FKPath}$, thereby identifying $T$ as sensitive, we augment $T$ with a PU hash column $\mathsf{pu}$.

\noindent\textbf{(2) Aggregation replacement $(\mathsf{Re}^{\mathsf{SIMD}}_{1,2}(Q))$.}
As indicated in  Lines~\ref{line:replaceaggregationstart}--\ref{line:replaceaggregationend}, 
each aggregation operator 
$\gamma_{G;\mathsf{Aggs}}(T)$ over a sensitive relation $T$ 
is replaced with its PAC counterpart. 
Formally, let 
$\mathsf{Aggs}=\{\mathsf{agg}_1,\ldots,\mathsf{agg}_d\}$. 
If $T$ contains the PU hash column $\mathsf{pu}$, we rewrite
$\gamma_{G;\mathsf{agg}_1,\ldots,\mathsf{agg}_d}(T)$
as
$\gamma_{G;\mathsf{pac\_agg}_1,\ldots,\mathsf{pac\_agg}_d}(T)$.

\noindent\textbf{(3) Selection rewriting in sub-expression $E$ ($\mathsf{Re}^{\mathsf{SIMD}}_{1,2,3}(Q)$).}
As shown in Lines~\ref{line:selectioncorrelatedsubquerystart}--\ref{line:selectioncorrelatedsubqueryend}, 
a selection of the form $\sigma_{\mathsf{p}}(T)$—where 
$\mathsf{p}$ depends on inner aggregation results is rewritten as follows. 
We first compute the vector-lifted Boolean expression 
$\vec{\mathsf{p}}$ and then apply
\[
\pi_{\mathsf{cols}(T)\setminus \{\mathsf{pu}\}, \mathsf{pac\_select}(\mathsf{pu}, \vec{\mathsf{p}}) \rightarrow \mathsf{pu}}(T)
\]
to update the hash column $\mathsf{pu}$. 
Specifically, for each position $1 \leq j\leq m$, the output bit of $\mathsf{pac\_select}$ is set to $1$ 
iff the corresponding tuple satisfies the predicate when evaluated on the aggregation result of the $j$-th sub-sample and the original $\mathsf{pu}[j] = 1$; otherwise, the output bit is set to $0$.

\noindent\textbf{(4) Final projection and noise addition $\mathsf{Re}^{\mathsf{SIMD}}_{1,2,3,4}(Q)$.}
As indicated in 
In Lines~\ref{line:addnoisestart}--\ref{line:addnoiseend},  the top-level projection operator $\pi$ is rewritten so as to evaluate  vector-lifted expressions over the aggregation results.  Subsequently, calibrated noise is injected into each resulting expression by applying $\mathsf{pac\_noised}(\cdot, j^*, P)$, where the secret index $j^*$  and the probability vector $P$ are sampled and initialized identically to the PAC-DB procedure described in Section~\ref{sec:packdbsemantics}. 
This yields a PAC-private output. We therefore define
\[ 
\mathsf{Output}_{\mathrm{SIMD\text{-}PAC\text{-}DB}}(\mathbb{I}, Q)
\;:=\;
\mathsf{Re}^{\mathsf{SIMD}}_{1,2,3,4}(Q).
\]

Below, we establish the computational equivalence between PAC-DB and SIMD-PAC-DB for every admissible set of base relations $\mathbb{I}$ and every supported query $q \in \mathcal{Q}$. To enable a formal comparison, we couple their randomness. Specifically, we assume that both engines rely on the same underlying randomness for subset generation. Concretely, the PAC-DB compiler instantiates $\mathsf{SamplePU}$ so that its sampled subsets $\{S_j\}_{j=1}^m$ coincide exactly with the implicit subsets induced by the SIMD-PAC-DB compiler; that is, $S_j =
\sigma_{\mathsf{pac\_hash}_K(pk)[j] = 1}(U),$ for each $j \in \{1,\dots,m\}$. 
Moreover, we assume that both approaches share the same sampled secret index $j^*$ and the same initial probability state $P$. Finally, we assume that PAC-DB, like SIMD-PAC-DB, augments each sensitive base relation with a hashed PU column. 

\begin{theorem}[\textbf{Equivalence of PAC-DB and SIMD-PAC-DB}] \label{thm:equivalence}
Let $\mathbb{I}$ be a set of base relations, and let $\mathbb{I}^* \subseteq \mathbb{I}$ denote the subset of sensitive relations, where each $I \in \mathbb{I}^*$ is augmented with a PU hash column. 
Let $Q \in \mathcal{Q}(\mathbb{I})$ be any supported query. For any fixed choice of $\mathsf{pac\_hash}$ and any fixed realization of both the secret randomness and the noise randomness, the PAC-DB and SIMD-PAC-DB compilers produce identical outputs. Formally,
\[
\mathsf{Output}_{\mathrm{SIMD\text{-}PAC\text{-}DB}}(\mathbb{I}, Q)
\;=\;
\mathsf{Output}_{\mathrm{PAC\text{-}DB}}(\mathbb{I}, Q).
\]
\end{theorem}

\textsc{Proof.}
We establish the equivalence in four steps.

\noindent
\textit{\textbf{1: Base case $F(\mathbb{T})$ with $\mathbb{T} = \mathbb{I}$.}}
Under the hash-based instantiation, the $j$-th sampled database in PAC-DB is defined as
$
\mathbb{I}^{(j)}
=
\{\sigma_{\mathsf{pu}[j]=1}(I) \mid I \in \mathbb{I}^*\}
\;\cup\;
(\mathbb{I} \setminus \mathbb{I}^*).
$
For each hash value $h$ of a PU identifier, define
$
\mathbb{I}_{h}
=
\{\sigma_{\mathsf{pu}=h}(I) \mid I \in \mathbb{I}^*\}
\;\cup\;
(\mathbb{I} \setminus \mathbb{I}^*).
$
Since $F$ is deterministic and does not combine tuples across distinct PU identifiers, we have
$
\sigma_{\mathsf{pu}=h}
\big(
\mathsf{Re}^{\mathsf{SIMD}}_1(F(\mathbb{I}))
\big)
=
F(\mathbb{I}_{h}).
$
Consequently, after propagating the PU transformation,
\begin{equation}\label{proofequation1}
\pi_{\{\mathsf{cols}\}\setminus\{\mathsf{pu}\}}
\Big(
\sigma_{\mathsf{pu}[j]=1}
\big(
\mathsf{Re}^{\mathsf{SIMD}}_1(F(\mathbb{I}))
\big)
\Big)
=
\pi_{\{\mathsf{cols}\}\setminus\{\mathsf{pu}\}}
\big(
F(\mathbb{I}^{(j)})
\big).
\end{equation}

\medskip
\noindent
\textit{\textbf{2: First aggregation.}}
By the definition of PAC aggregation, applying $\mathsf{pac\_agg}$ to 
$\mathsf{Re}^{\mathsf{SIMD}}_1(F(\mathbb{I}))$
produces a vector $(r_1,\ldots,r_m)$, where each component $r_j$ equals the result of applying the scalar aggregation $\mathsf{agg}$ to the left-hand side of Eq.~\eqref{proofequation1}. Equivalently,
$r_j = \mathsf{agg}\big(F(\mathbb{I}^{(j)})\big)$.
Therefore,
\begin{equation}\label{eq:agg_equiv}
\mathsf{Re}^{\mathsf{SIMD}}_{1,2}
\Big(
\gamma_{G;\mathsf{agg}}\big(F(\mathbb{I})\big)
\Big)
=
\gamma_{G;\mathsf{List}(\alpha)}
\left(
\uplus_{j=1}^{m}
\gamma_{G;\mathsf{agg}\rightarrow\alpha}
\big(
F(\mathbb{I}^{(j)})
\big)
\right).
\end{equation}

\medskip
\noindent
\textit{\textbf{3: Selection and nested aggregation.}}
Consider sub-expression~(a):
\[
E
=
\pi_{\mathsf{cols}(T_s^2)}
\!\left(
\sigma_{\mathsf{p}}
\big(
T_s^2
\;\mathsf{Join}_{G_1}\;
\gamma_{G_1;\mathsf{Aggs}_1}(T_s^1)
\big)
\right).
\]
Assume that both $T_s^1$ and $T_s^2$ are produced by subqueries in $F(\mathbb{I})$, i.e. $T_s^1 = F_1(\mathbb{I})$ and  $T_s^2 = F_2(\mathbb{I})$.
By Eq.~\eqref{eq:agg_equiv}, inner aggregation is preserved under sampling.
The join operator must satisfy the PAC Link constraint, namely that $G_1$ forms a foreign key from $T_s^2$ to $\gamma_{G_1;\mathsf{Aggs}_1}(T_s^1)$.
Under this condition, the join is key-preserving on $T_s^2$: it introduces no duplication and merely appends aggregation attributes.
Thus, in PAC-DB,
\[
E^{(j)}
=
\pi_{\mathsf{cols}(F_2(\mathbb{I}^{(j)}))}
\!\left(
\sigma_{\mathsf{p}}
\big(
F_2(\mathbb{I}^{(j)})
\;\mathsf{Join}_{G_1}\;
\gamma_{G_1;\mathsf{Aggs}_1}(F_1(\mathbb{I}^{(j)}))
\big)
\right)
\]
returns exactly those tuples in $T_j = F_2(\mathbb{I}^{(j)})
\;\mathsf{Join}_{G_1}\;
\pi_{G_1} (\gamma_{G_1;\mathsf{Aggs}_1}(F_1(\mathbb{I}^{(j)})))$ whose aggregation results satisfy $\mathsf{p}$. By Eq.~\eqref{proofequation1},
this relation is equivalent to selecting precisely those tuples in $ T=
F_2(\mathbb{I})
\;\mathsf{Join}_{G_1}\;
\pi_{G_1}\!\left(
\gamma_{G_1;\mathsf{Aggs}_1}\big(F_1(\mathbb{I})\big)
\right)$
whose aggregation results satisfy $\mathsf{p}$ and whose corresponding
PU identifiers are selected in sub-sample $j$.

In SIMD-PAC-DB, this consistency is enforced by
\[
\mathsf{Re}^{\mathsf{SIMD}}_{1,2, 3}(E) = \sigma_{\mathsf{pu} \neq 0}
\Big( 
\pi_{\mathsf{cols}\setminus\{\mathsf{pu}\},\mathsf{pac\_select}(\mathsf{pu},\vec{\mathsf{p}})
\rightarrow \mathsf{pu}}
\big(
\mathsf{Re}^{\mathsf{SIMD}}_{1,2}(E)
\big)
\Big).
\]
Here, the vector-lifted predicate $\vec{\mathsf{p}}$
is computed such that $\vec{\mathsf{p}}[j]$
indicates whether the tuple satisfies $\mathsf{p}$ with respect to
the $j$-th aggregation result $r_j$, where
$(r_1,\ldots,r_m)$ denotes the vector of aggregation results
produced by the inner PAC aggregation.
The operator $\mathsf{pac\_select}$ then performs a component-wise logical
AND between $\vec{\mathsf{p}}$ and the original
PU-hash bit vector, yielding an updated $\mathsf{pu}$.
Consequently, $\mathsf{pu}[j]$ is set to $1$ if and only if the tuple occurs in $T_j$ (i.e. sub-sample $j$) and satisfies $\mathsf{p}$ when evaluated using the aggregation result computed for sub-sample $j$. 
The outer selection $\sigma_{\mathsf{pu} \neq 0}$ removes all tuples that do not satisfy the predicate in any of the $m$ sub-samples.
Hence,
\begin{equation}\label{equationsubexpression1}
\pi_{\mathsf{cols}\setminus\{\mathsf{pu}\}}
\Big(
\sigma_{\mathsf{pu}[j]=1}
\big(
\mathsf{Re}^{\mathsf{SIMD}}_{1,2,3}(E)
\big)
\Big)
=
\pi_{\mathsf{cols}\setminus\{\mathsf{pu}\}}
\big(
E^{(j)}
\big).
\end{equation}

The relation $E$ may subsequently serve as an input to an outer query of type $F$. A noticeable operation within $F$ is the join between $E$ and another PU-linked relation, denoted $A$.
By the PAC Link constraint, tuples may be combined only if they lead to the same PU row. Moreover, the join must be executed along the (acyclic) foreign-key chain leading to the PU table $U$. Along this linear path, the relation $A$ inherits the updated hash column \emph{pu} propagated from its inner subtree.
Consequently, no additional transformation of the join operator is required. 
Therefore, Eq.~\eqref{proofequation1} naturally extends to queries of the form $F(\mathbb{I} \cup \{E\})$.

By iterating this argument, Eq.~\eqref{equationsubexpression1}
continues to hold for expressions $E'$ where
$T_s^1$ and $T_s^2$ are generated by $F(\mathbb{I} \cup \{E\})$.
More generally, let $\mathcal{T}$ be any query obtained by
arbitrarily nesting $F$ and sub-expression~(a).
Then
\begin{equation}\label{equationft}
\pi_{\mathsf{cols}\setminus\{\mathsf{pu}\}}
\Big(
\sigma_{\mathsf{pu}[j]=1}
\big(
\mathsf{Re}^{\mathsf{SIMD}}_{1,2,3}(\mathcal{T}(\mathbb{I}))
\big)
\Big)
=
\pi_{\mathsf{cols}\setminus\{\mathsf{pu}\}}
\big(
\mathcal{T}(\mathbb{I}^{(j)})
\big).
\end{equation}

Propagating Eq.~\eqref{equationft} to the top-level aggregation yields
\[
\mathsf{Re}^{\mathsf{SIMD}}_{1,2,3}
\Big(
\gamma_{G;\mathsf{agg}}\big(\mathcal{T}(\mathbb{I})\big)
\Big)
=
\gamma_{G;\mathsf{List}(\alpha)}
\left(
\uplus_{j=1}^{m}
\gamma_{G;\mathsf{agg}}
\big(
\mathcal{T}(\mathbb{I}^{(j)})
\big)
\right).
\]
Here, $\mathsf{pac\_agg}(\mathsf{pu}, \cdot)$ in SIMD-PAC-DB produces a vector whose $j$-th component is obtained by aggregating exactly those tuples with $\mathsf{pu}[j] = 1$. This condition indicates that the tuples belong to the $j$-th hash-induced sub-sample and satisfy all propagated predicates within that sub-sample.

We note that vector-valued predicates $\mathsf{p}$ are evaluated component-wise using Eq.~\eqref{vectorexpressionlambda}.
The function $\mathsf{list\_zip}$ combines vectors
$\mathsf{col}_1,\ldots,\mathsf{col}_a$
into a vector $\mathsf{v}$ with
$\mathsf{v}[j]=(\mathsf{col}_1[j],\ldots,\mathsf{col}_a[j])$.
Applying $\mathsf{list\_transform}$ with
$\lambda x.\,\mathsf{e}(x[1],\ldots,x[a])$
ensures coordinate-wise evaluation, preserving per-sample semantics.

\medskip
\noindent
\textit{\textbf{4: Projection and noise addition.}}
Post-aggregation projections are vector-lifted via Eq.~\eqref{vectorexpressionlambda}.
Both systems apply the same noise mechanism with fixed randomness. 
Let $T' 
=
\gamma_{G;\mathsf{List}(\alpha)\rightarrow\alpha'}
\left(
\uplus_{j=1}^{m}
Q(\mathbb{I}^{(j)})
\right), $ where $Q(\mathbb{I}^{(j)}) = \pi_{\mathsf{e}}
\big(
\gamma_{G;\mathsf{agg}\rightarrow\alpha}(\mathcal{T}(\mathbb{I}^{(j)}))
\big).$
Then
\begin{equation}\label{pronoiseequation}
\mathsf{Re}^{\mathsf{SIMD}}_{1,2,3,4}
\!\left(
\pi_{\mathsf{e}}\big(\gamma_{G;\mathsf{agg}}(\mathcal{T}(\mathbb{I}))\big)
\right)
=
\pi_{\mathsf{pac\_noised}(\alpha',j^*,B)}
\big(T'
\big).
\end{equation}

The right-hand side of Eq.~\eqref{pronoiseequation} is exactly the output of PAC-DB, while the left-hand side equals the output of SIMD-PAC-DB.
The general case with multiple aggregations and projections follows component-wise.
Therefore,
\[
\mathsf{Output}_{\mathrm{SIMD\text{-}PAC\text{-}DB}}(\mathbb{I},Q)
=
\mathsf{Output}_{\mathrm{PAC\text{-}DB}}(\mathbb{I},Q).\quad\qed
\]

\medskip
As an immediate corollary of Theorem~\ref{thm:equivalence}, the privacy guarantee established in Theorem~\ref{thm:post_adv} also holds for SIMD-PAC-DB.

\section{SIMD-optimized PAC Aggregates}
\label{sec:simd}
Our stochastic aggregate functions \texttt{pac\_\{count,min,max,sum,avg\}} compute an array of 64 aggregates, in which the value at position 0$\leq\!\!j\!\!<$64 accumulates contributions from those rows whose \texttt{key\_hash} has bit~$j$ set. We now describe a number of optimizations for them. For portability and future-proofness, we avoid writng explicit SIMD, but rather rely on C++ compiler auto-vectorization.

\vspace*{2mm}\noindent{\bf Na\"ive SIMD-unfriendly}.
 We started with an implementation that updates in the loop with \texttt{if ((key\_hash$>\!\!>$j)\&1) update(aggr[j])}. However, using an \texttt{if} inhibits the use of SIMD and also suffers from CPU branch-mispredictions, especially because the probability to take the branch is 50\% in PAC, and thus is unpredictable, the worst case.
Therefore, the \texttt{PacCountUpdate()} code shown on page 1 already used {\bf predication}, that transforms a control- into a data-dependency: incrementing \texttt{count[j]} by the boolean condition (0 or 1). For \texttt{SUM/AVG} we achieve this by multiplying the value to add with (0 or 1).
\texttt{MIN/MAX} subtract 1 from the boolean so it becomes (-1, 0), which is used as a mask to select either the aggregate or the new value.
We use \texttt{union} types to interpret values as signed/unsigned integers or floating-point, depending on the operation (mask vs comparison); this works well with auto-vectorization. 

\vspace*{2mm}\noindent{\bf Cascading.}
SIMD instructions execute a single operation on multiple data \emph{lanes}. Lane-widths are 8, 16, 32, or 64-bits. 
The thinner the lane, the more data items can be processed in one instruction.
The DuckDB \texttt{COUNT()} returns a \texttt{UBIGINT}, which requires slowest lane width (64-bits). In order to profit more from SIMD we use \emph{Cascading} by keeping two arrays: 
\texttt{uint8\_t probabilistic\_total8[64]}, and \texttt{uint64\_t probabilistic\_total[64]}. Every 255 updates, we add total8 to total, and zero the former.
This prevents overflow, but still allows to perform the additions in a lane-width of 8-bits.

Autovectorizing compilers strongly prefer computations on a single lane-width.
This is a complication, because \texttt{key\_hash} is 64-bits.
The cascading count hence cannot use 8-bits lanes, but must use 64-bits lanes.
Therefore, we use {\bf SWAR (SIMD Within A Register)} to perform the predicated add: we change into
\texttt{{\bf\color{red} uint64\_t} probabilistic\_total8[{\bf\color{red} 8}]} that consists of  8x "virtual" 8-bits lanes.
Figure~\ref{fig:godbolt} just shows the SWAR update kernel, and leaves out the propagation of the 8-bits total8 values towards the full totals.\footnote{Our code supports a {\tt PAC\_GODBOLT} define to extract the aggregate kernels e.g. {\tt cpp -DPAC\_GODBOLT -P -E -w  src/include/aggregates/pac\_count.hpp} which can be copy-pasted in godbolt.org. Works also for pac\_sum.hpp and pac\_min\_max.hpp }

\begin{figure}
  \centering
  \includegraphics[width=\linewidth]{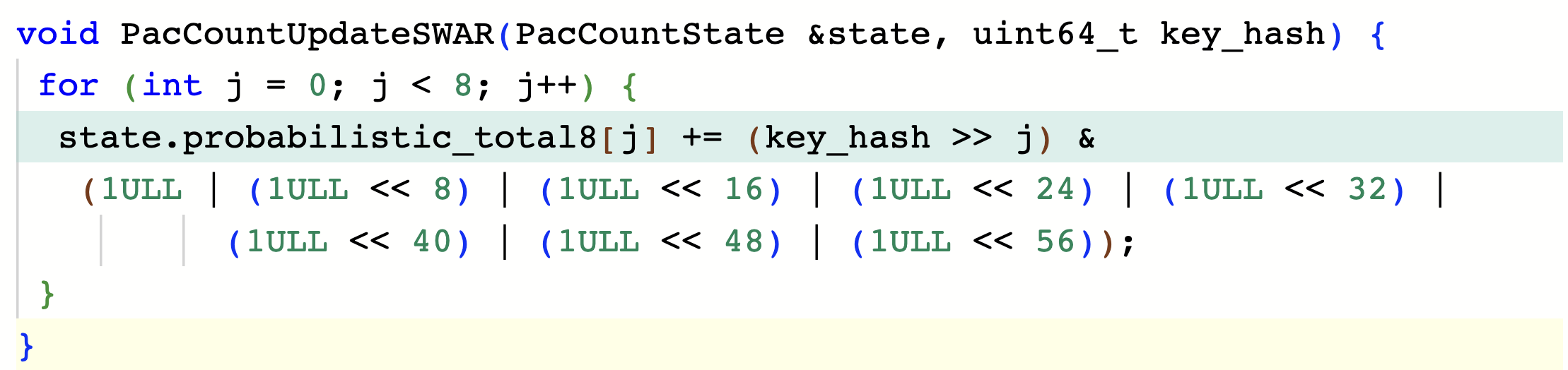}
  \includegraphics[width=\linewidth]
  {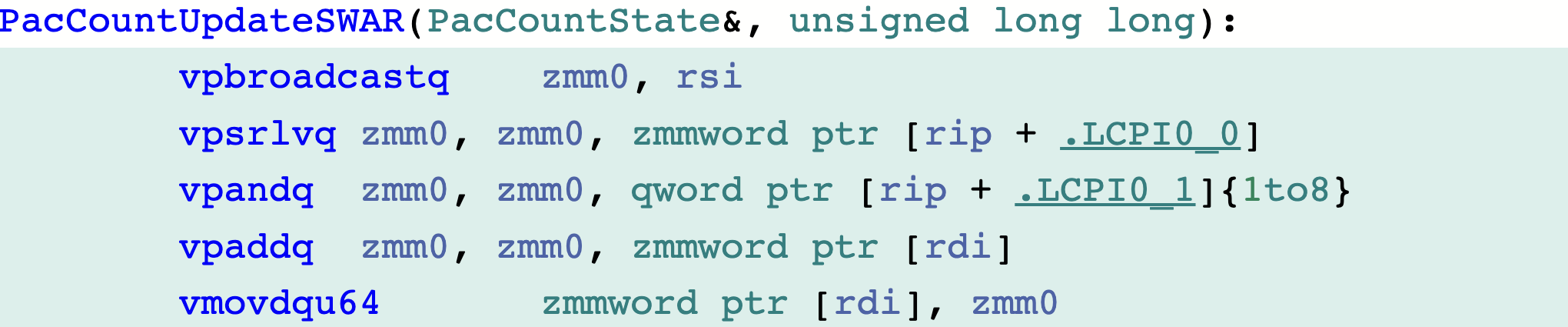}
   \includegraphics[width=\linewidth]{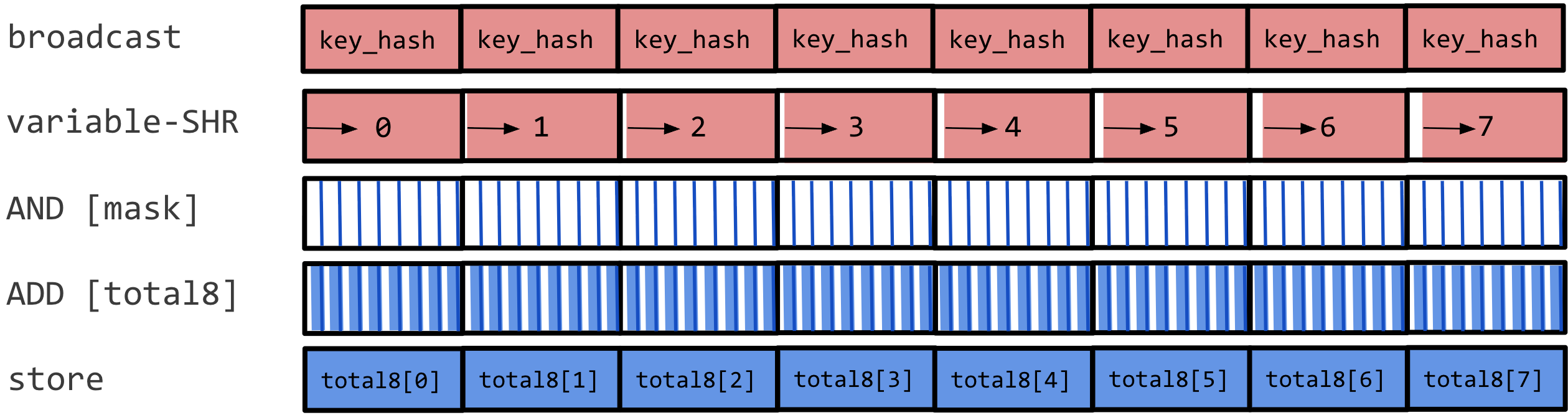}
\vspace*{-7mm}
  \caption{godbolt.org screenshots of  \texttt{PacCountUpdateSWAR()} (top), that updates 64-bits \texttt{probabilistic\_total8}, which consists of 8x "virtual" 8-bits lanes. It uses one AND that masks 8 bits and the ADD adds these (0 or 1)s to 8x 8-bits numbers inside the 64-bits integer. The loop was reduced to 8 iterations (down from 64) because of this SWAR (SIMD Within A Register). Auto-vectorization makes the loop disappear: SIMD+SWAR. Five AVX512 instructions update 64 counters!} 
\label{fig:godbolt}\vspace*{-4mm}
\end{figure}


We also applied SWAR+Cascading to \texttt{SUM/AVG} . Their state keeps five arrays of 64 values: forming five "levels" of  8, 16, 32, 64 resp. 128 bits aggregates.
Their update kernel decides based on the magnitude of the new value, which level to add the value to. It keeps for each level an exact sum, and if the new value makes that exact sum overflow, it is reset to 0 and the level is flushed before doing the update. The SUM/AVG implementations are identical, both also keep an exact update count; AVG just divides all 64 aggregates by this count at finalization. 

\vspace*{2mm}\noindent{\bf Buffering.}
The total state for cascading SUM/AVG is large: 64*(1+2+ 4+8+16)$\simeq$2KB. DuckDB allows aggregates to allocate memory at runtime, and we use this for {\bf lazy allocation} of the cascading levels. DuckDB's state-of-the-art relational Aggregate operator~\cite{kuiperaggr}, where our stochastic aggregation extension functions hook into, uses a parallel and adaptive algorithm, whose first phase consists of thread-local aggregation that also partitions the data. It creates small hash tables which attempt to eliminate duplicates and aggregate values immediately (early aggregation). However, when these hash tables grow large, it stops inserting into them and starts over, creating a new hash table and partitions, allowing the previous partitions to be swapped to disk. In the second phase of the algorithm, each thread gets assigned all matching partitions and then \emph{combines} them and then \emph{finalizes} the aggregates. 

To limit memory usage, we implemented {\bf buffering}, where a small state record just has space for three \texttt{(value,key\_hash)} pairs, and the \emph{update} method only appends a new value there. On receiving the fourth value, the real Cascading state is allocated on need, and the buffer gets flushed, updating the aggregate with four values in batch. 
Buffering ensures that our aggregate arrays are only allocated when values actually get aggregated together.

\vspace{2mm}\noindent{\bf Approximation.}
Sums over large tables can grow very large, and therefore often end up requiring the largest DuckDB numeric datatype \texttt{HUGEINT} (128-bits).
This datatype is not natively supported by CPUs and therefore slow to compute on. It makes our state very large (64x16=1KB) and also leaves no room for SIMD execution. 
However, given the fact that our stochastic aggregates will eventually be noised, and the noise is at least a few percent of the magnitude (given MI=1/128), we do not need to compute exact aggregates in the first place.
A simple strategy that SIMD-PAC-DB follows for summing \texttt{DOUBLE} values is to use \texttt{FLOAT} state only.

For integers, we use 16-bit sums, organized in {\bf 25 lazily allocated levels} which cascade every 4 bits.
An incoming value $v$ is routed to level~$\ell(v) = \lfloor(\mathit{msb}(|v|) - 8)/4\rfloor$ based on its most significant bit.  On overflow, only the upper 12~bits of each counter cascade: $C_{k+1}[j] \gets C_{k+1}[j] + (C_k[j] \gg 4)$, so the worst-case relative error is $2^{-12} \approx 0.024\%$ --- negligible in comparison PAC noising. 

We use 16-bits counters because we want to use an as thin lane as possible, but 8-bits counters would provide too little precision, and use quite many (25) per 4-bits staggered counters  because staggering every 8-bits would similarly lose significant precision. We limit to 25 counters because this covers 24*4+16=112 bits, which is enough for summing large columns of (64-bits) \texttt{BIGINT}.

The Approximation state contains 25 pointers (200 bytes) to dynamically allocated level arrays. In memory-heavy aggregations with many ($K$) distinct keys and $N$ rows, many pointers will stay null, because each aggregate will need just $\sim$log($N/K$)/4 levels, and $N/K$ is small if $K$ is large.
Therefore we re-use the back part of this pointer array, as long as it is still empty, to store the first level we allocate (64x16-bits=128 bytes).
This optimization we call {\bf inlining}.
In distributions with large $K$, close to $N$, level flushes are rare such that each aggregate will allocate typically one level, reducing memory footprint to almost the minimum (200$\simeq$128 bytes). 

For approximated signed integers SUM/AVG, we decided to keep a complete second state for negative numbers, stored negated. 
At finalization we subtract the two sides:
$\mathit{result}[j] = 2\cdot(C_{\text{pos}}[j] - C_{\text{neg}}[j])$.
This {\bf Two-sided SUM} helps maintaining the precision of the approximate sums in distributions with a mix of positive and negative numbers, that could otherwise cancel each other. 
Columns with a signed type like \texttt{DECIMAL} often contain positive numbers only; and because we lazily allocate the negated side state, we typically now get unsigned performance on signed types.

\vspace{2mm}\noindent{\bf Pruning.}
\texttt{pac\_[noised\_]min/max} maintain a running global bound $g$ --- the worst extreme of all 64 counters ($g = \min_j E[j]$ for MAX; $g = \max_j E[j]$ for MIN).  Incoming values that cannot improve $g$ skip the update entirely.  The bound is refreshed every 2048 updates to amortize cost.  This optimization works on all distributions except monotonically increasing (MAX) or decreasing (MIN) ones -- these are adversarial because the aggregate changes on every update.

\vspace*{2mm}\noindent{\bf Diversity Check.} For NULL handling, all aggregates maintain an accumulator that gets XOR-ed with \texttt{key\_hash} on every update; at finalization its 0s indicate in which worlds it never received an update.
We exploit this accumulator to throw a runtime error at finalization if (i) the aggregate was updated many times but (ii) close to 32 worlds never received an update. This indicates it was only updated with a single \texttt{key\_hash}. This happens when one does GROUP BY \emph{pu} or something strongly correlated to it.
While such queries should be rejected by the compiler, all aggregate functions perform this check for extra safety.

\section{Evaluation}
\label{sec:evaluation}
We evaluate DuckDB-PAC along performance and utility, identifying the following research questions:

\begin{description}
  \item[RQ1] What is the performance of our new stochastic aggregation functions, and how do the optimizations proposed in Section~\ref{sec:simd} impact this?
  \item[RQ2] How much overhead do PAC aggregates and joins add relative to default query execution in analytical benchmarks?
  \item[RQ3] How do we define the quality of the approximate query answers, and how good is it in analytical benchmarks?
  \item[RQ4] How can we characterize the current SQL coverage of DuckDB-PAC, and what would be the next steps to enlarge it?
\end{description}
\noindent
We benchmark on TPC-H, ClickBench, and SQLStorm, and run all experiments on four machines, with RAM size in GB:
\begin{itemize}
  \item MacBook Pro -- 12-core ARM M2 Max, 32GB, SSD.
  \item c8gd.4xlarge -- 16-core ARM AWS Graviton4, 32GB, SSD.
  \item c8i.4xlarge  -- 16-core x86 Intel Granite Rapids, 32GB, io2.
  \item c8a.4xlarge  -- 16-core x86 AMD EPYC 9R45, 32GB, io2.
\end{itemize}

All experiments use DuckDB v1.4.2 with the DuckDB-PAC extension compiled against the same source tree, with Clang 22.1.0. 

\begin{figure}[ht]
  \centering
  \includegraphics[width=\linewidth]{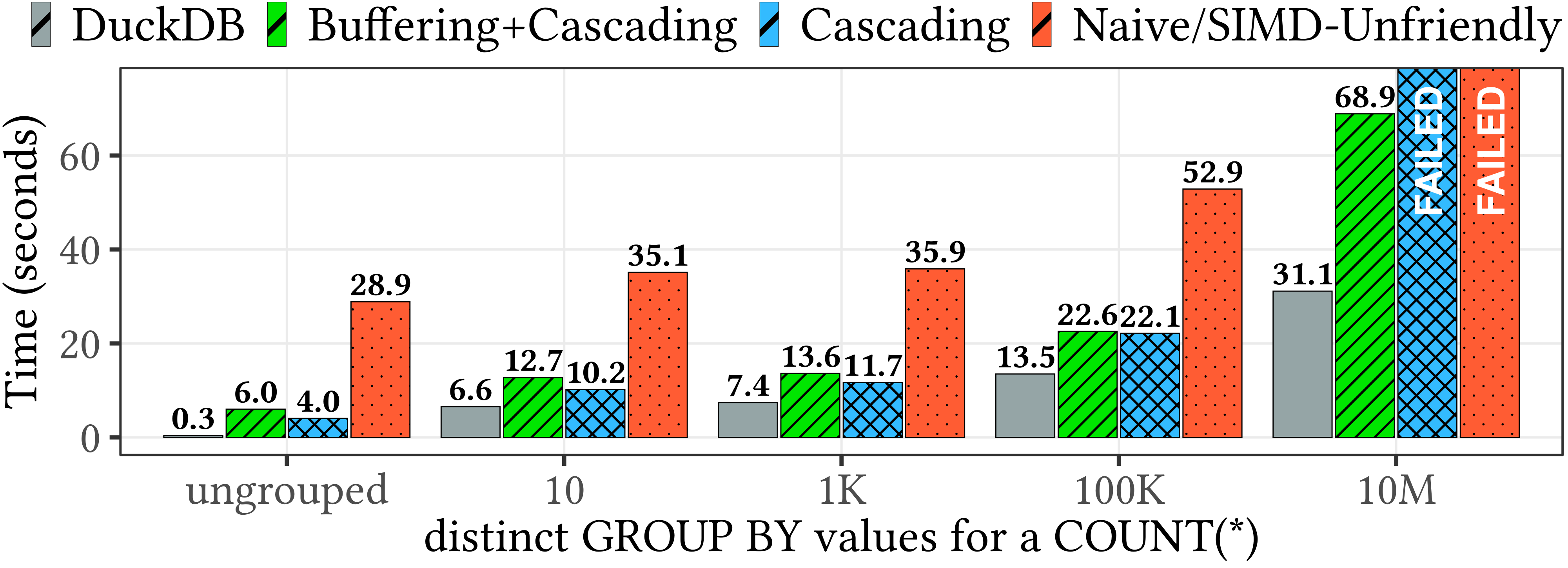}\vspace*{-4mm}
  \caption{\textbf{PAC Count optimization impact on Intel (Granite Rapids), 1G rows, scattered groups.} Cascading 8-bit into 64-bit counters enables SIMD improvements over the Na\"ive approach. Buffering adds little overhead, but reduces RAM footprint with many groups --- the others OOM at 10M.}
  \label{fig:count}\vspace*{-4mm}
\end{figure}

\begin{figure}[ht]
  \centering
  \includegraphics[width=\linewidth]{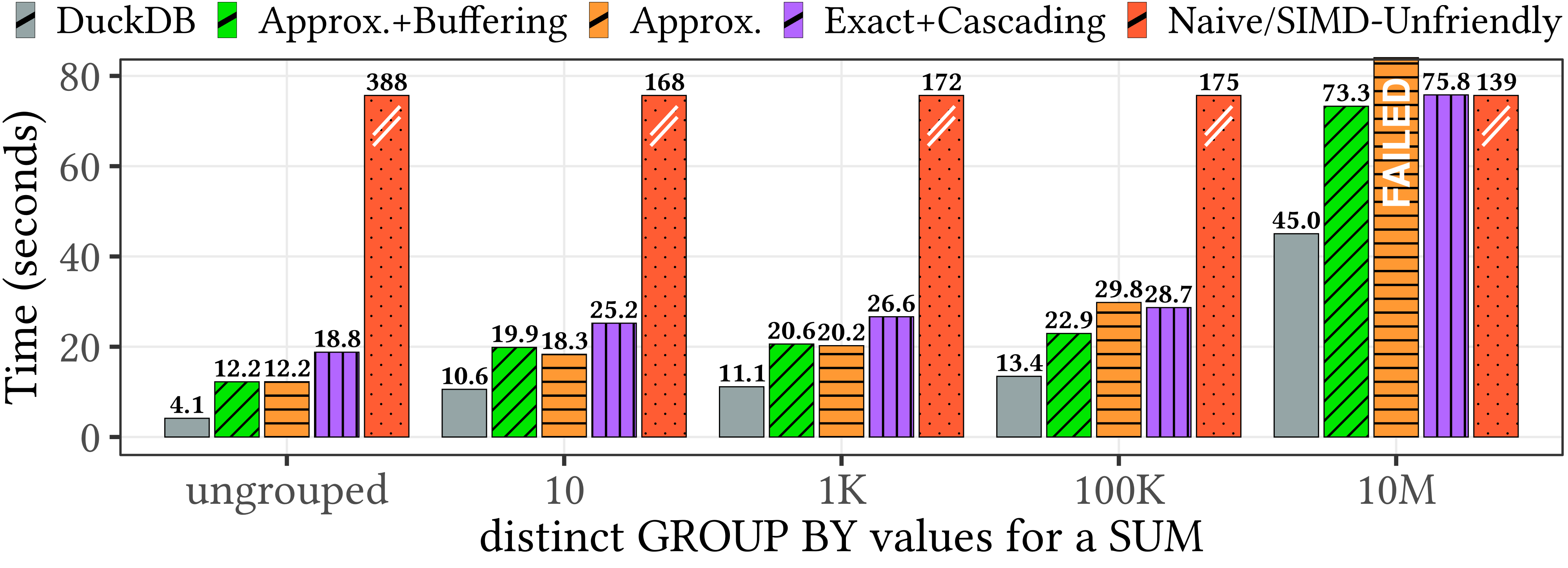}\vspace*{-4mm}
  \caption{\textbf{PAC Sum optimization impact on Macbook, 1G rows, scattered groups. The Approximate sum with 25 lazy counter levels is faster than the Cascading SUM, because it only operates in thin 16-bit SIMD lanes. Buffering adds little overhead but  prevents OOM in case of many groups (10M).}}
  \label{fig:sum}\vspace*{-4mm}
\end{figure}

\begin{figure}[ht]
  \centering
  \includegraphics[width=\linewidth]{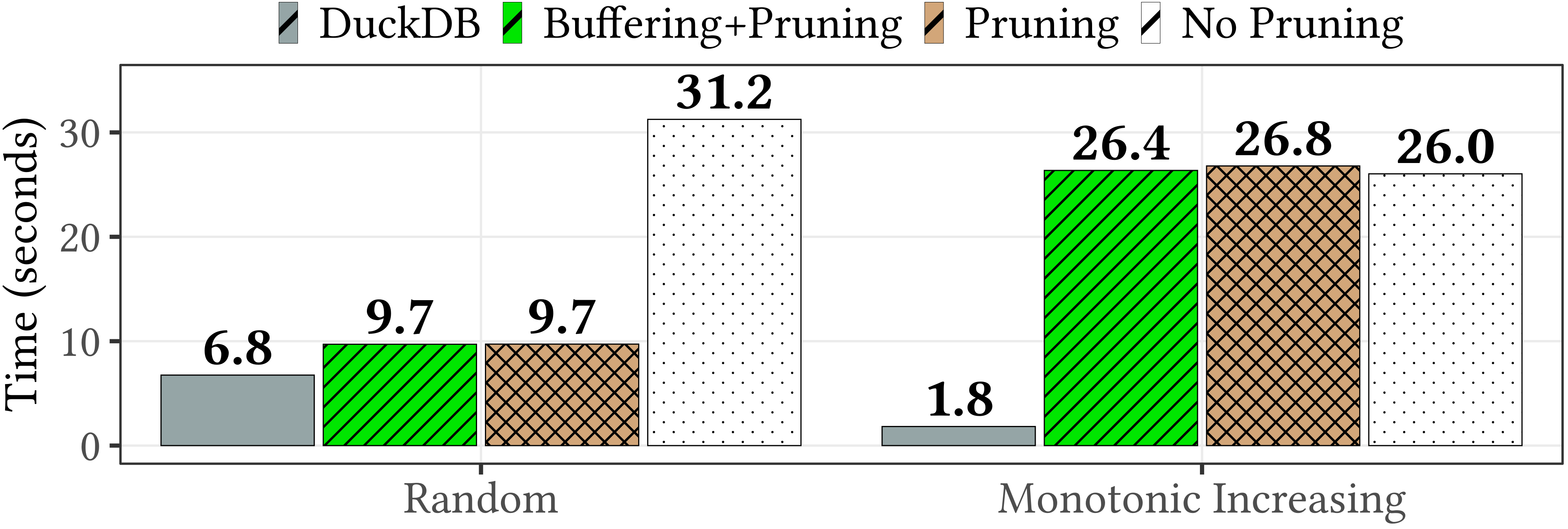}\vspace*{-4mm}
  \caption{\textbf{PAC MAX optimization impact on Graviton4, 1G rows, ungrouped. Bound pruning (green) improves \texttt{pac\_max()} performance $3\times$ on random data (white) and within only $1.4\times$ of standard DuckDB MAX (grey). On monotonically increasing data, all variants converge.}\vspace*{-5mm}
}\label{fig:max}
\end{figure}

\subsection{Aggregation Optimizations (RQ1)}
\label{sec:microbench}

We generate synthetic tables in DuckDB with row counts $N \in \{10M, 100M, 1G\}$
and vary the value distribution (uniform random, monotonically increasing, monotonically decreasing), domain size (tiny: $0$--$100$; small: $0$--$10^4$; medium: $0$--$10^5$; large: $0$--$10^6$).
We aggregate ungrouped, and with $K \in \{ 10, 1K, 100K, 10M\}$ distinct GROUP BY keys, which appear scattered or consecutively.

Figure~\ref{fig:sum} shows micro-benchmark performance of the query:\\
\texttt{\small SELECT pac\_noised\_sum(hash,val) FROM tab\_1G [GROUP BY key]}\\ and Figure~\ref{fig:count} similarly for \texttt{\small pac\_noised\_count(hash)} . These queries run as SQL in the DuckDB shell with our extension loaded, and were configured with \texttt{\small SET threads=4; SET max\_memory='30G'; SET max\_temp\_directory\_size='90GB'}. The DuckDB aggregation runs in parallel, doing \texttt{state.update(value)} calls to the aggregates. Given the random order of the keys, all threads will see all groups, and thus have thread-local aggregation state for all groups. When all rows have been seen, the states are combined by a thread that calls \texttt{state.combine(other)} on all aggregate states for the same key, consolidating them into a single \texttt{state}. Eventually \texttt{state.finalize()} is called. Our PAC fused \texttt{pac\_noised\_..()} aggregates, then calls \texttt{pac\_noised(state)} to generate a single noised answer.

\begin{table}[t]
\centering
\caption{Approximate SUM accuracy: Two-Sided SUM wins.}\label{tab:sum-stability}\vspace*{-3mm}
\footnotesize
\renewcommand{\arraystretch}{0.85}
\begin{tabular}{lrrrrrr}
\toprule
& \multicolumn{3}{c}{\textbf{Single Sum}} & \multicolumn{3}{c}{\textbf{Two-Sided Sum (pos/neg)}} \\
\cmidrule(lr){2-4} \cmidrule(lr){5-7}
\textbf{Distribution} & \textbf{\% Err} & $z^2$ & \textbf{$\sigma^2$ ratio} & \textbf{\% Err} & $z^2$ & \textbf{$\sigma^2$ ratio} \\
\midrule
all\_same         & 0.16 &  2.52 &  1.00 & 0.07 & 0.51 & 0.97 \\
bimodal           & 0.20 &  2.18 &  0.97 & 0.11 & 0.61 & 0.99 \\
exponential       & 0.30 &  1.95 &  0.92 & 0.14 & 0.47 & 0.99 \\
negative\_mixed   & 113.30 & 13.31 & 209.88 & 22.37 & 0.004 & 0.99 \\
sparse\_large     & 0.10 &  0.0007 & 1.00 & 0.05 & 0.0002 & 1.00 \\
uniform\_bigint   & 0.21 &  3.74 &  0.97 & 0.12 & 1.22 & 1.03 \\
uniform\_int      & 0.21 &  3.58 &  1.01 & 0.12 & 1.12 & 1.00 \\
uniform\_smallint & 0.24 &  5.30 &  0.97 & 0.14 & 1.82 & 1.00 \\
uniform\_tinyint  & 0.14 &  1.98 &  0.92 & 0.07 & 0.52 & 1.01 \\
zipf\_like        & 0.13 &  0.0002 & 1.00 & 0.05 & 0.0  & 1.00 \\
\bottomrule
\end{tabular}\vspace*{-4mm}
\end{table}
\begin{figure}[t]
  \centering
  \includegraphics[width=\linewidth]{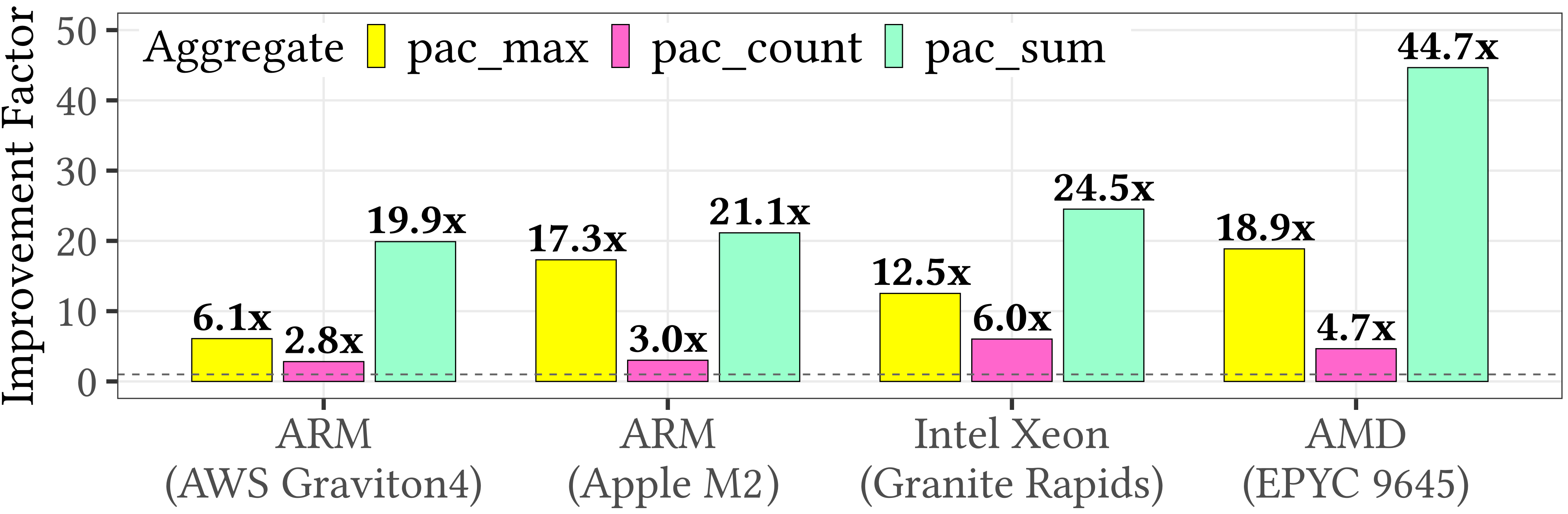}\vspace*{-4mm}
  \caption{\textbf{The impact of SIMD optimizations is largest on SUM and AMD (100M rows, ungrouped, avg of all column types).
}}\label{fig:simd_improvements}\vspace*{-9mm}
\end{figure}

\vspace*{2mm}\noindent{\bf Na\"ive SIMD-unfriendly.} 
The COUNT and SUM micro-benchmarks (Figures~\ref{fig:count} and \ref{fig:sum}) show slow performance for the \texttt{if-then} approach on the plain datatype, with the biggest disadvantage in SUM. Figure~\ref{fig:simd_improvements} confirms this happens across all platforms. 
The reason SUM is affected most is that DuckDB returns a \texttt{HUGEINT} result there, and arithmetic on such 128 bits is slow -- causing significant cost given that \texttt{pac\_sum} must update on average 32 such counters.
Another reason for the slow performance is branch mispredictions.

\vspace*{2mm}\noindent{\bf Cascading.} 
Our standard build generates portable 256-bit AVX2 on x86 instead of AVX512, which results in 10 assembly instructions instead of the 5 shown in Figure~\ref{fig:godbolt} --- however, this performs equally well: both Intel and AMD machines are able to retire 2x more AVX2 instructions as AVX512 in this kernel. Performance with 10 GROUP BY keys is lower than ungrouped because of the DuckDB hash table effort in the latter case. At 100K keys, performance noticeably drops, likely due to cache misses (the Cascading COUNT state size is 64*(8+1)=576 bytes, so each thread allocates 58MB of state). 
At 10M groups, every group will be updated just 100 times (1G rows).
We found that, due to the swapping of hash table partitions, the \texttt{state.update(val)} call will not find duplicates and hence creates a new state for every value. The need to allocate 576GB causes Cascading to fail with Out Of Memory (OOM), just like Na\"ive.

\vspace*{2mm}\noindent{\bf Buffering.} 
In the GROUP BY $K$=10M case, DuckDB finds the duplicate keys in the second phase of Aggregation, when threads call \texttt{state.combine(other)}. The Buffering approach delays the allocation of arrays for 64 aggregates until then. The fixed part of the Buffering-optimized aggregate state, allocated in the aggregation hash tables, is only 32-40 bytes. The need to allocate 32GB still makes DuckDB swap to disk, slowing it down, but it does not fail. Buffering is important to avoid OOMs, given the necessarily large state of all our PAC aggregates.
Figures~\ref{fig:count} and \ref{fig:sum} show that Buffering incurs little overhead in the ungrouped case, and when the amount $K$ of distinct GROUP BY values is small. But for SUM at $ K=100K$, it provides an advantage because it updates batches of four values at a time into the same aggregate. Therefore, CPU cache misses incurred when accessing the state are amortized more effectively.
\newpage

\vspace*{0mm}\noindent{\bf Pruning.} 
The pruning optimization in MIN/MAX  makes other performance optimizations in MIN/MAX largely irrelevant, and brings the performance of the \texttt{pac\_noised\_max(hash, val} within 50\% of DuckDB's \texttt{max(val)}; unless the data distribution is explicitly adversarial (monotonically increasing, for MAX).

\vspace*{2mm}\noindent{\bf Approximation.} 
Figure~\ref{fig:sum} shows that Cascading SUM is 50\% slower than Approximate in the ungrouped case.
An important question is how the SUM Approximation affects accuracy: 
Table~\ref{tab:sum-stability} shows accuracy across 10  distributions (1M values each).
The unnoised \texttt{pac\_sum(hash, val)} returns 64 sums, making it a handy tool to quickly compute 64 random half-samples for statistics.
We compute $z^2 = \text{RMSE}^2 / \text{Var}(\text{approx})$: values well below 1 mean the approximation noise is negligible compared to the statistical noise already present in each random sample --- where downstream confidence intervals remain valid. The variance ratio $\text{Var}(\text{exact}) / \text{Var}(\text{approx})$ is a sanity check: it should be close to 1, confirming that the approximation preserves the natural spread of sample totals.
The original signed counters fail on mixed-sign data (``negative\_mixed''): when positive and negative values cancel within the counter hierarchy, approximate totals collapse to zero, killing
 variance (var\_ratio = 210) and inflating $z^2$ to 13.3. The {\bf Two-sided Sum} technique that we adopted eliminates cancellation ($z^2$ drops to 0.004, var\_ratio normalizes to ~1). 
As a bonus, unsigned counters gain one bit of precision, improving percentage error across all distributions.

\begin{figure}[t]
    \centering
    \includegraphics[width=\linewidth]{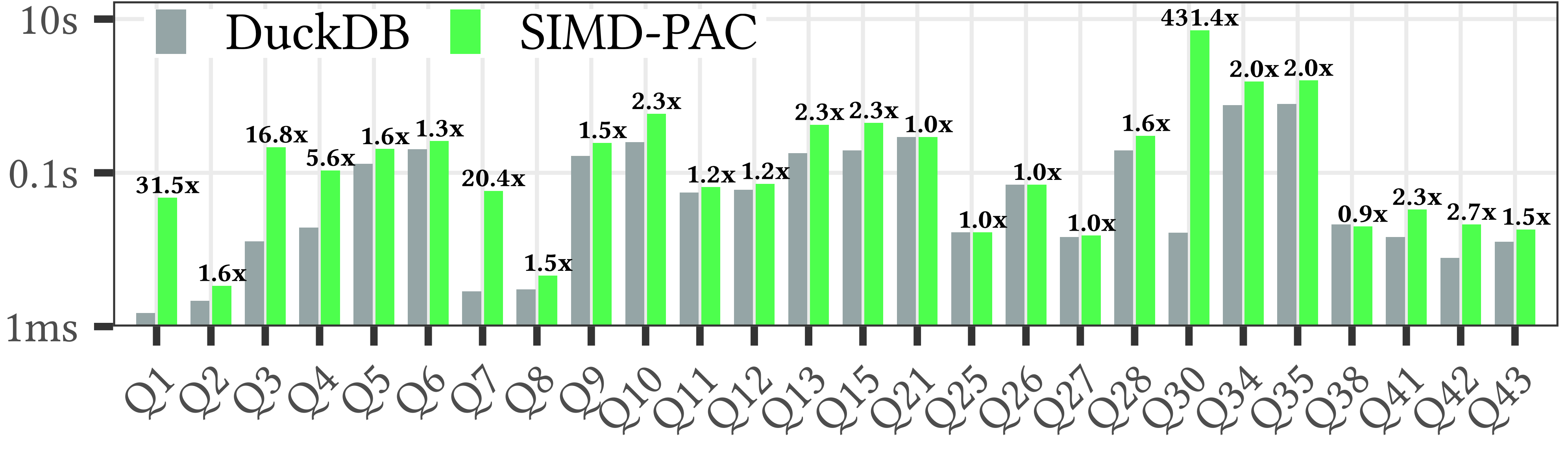}
    \vspace*{-9mm}
    \caption{\textbf{ClickBench query runtimes: DuckDB-PAC rejects 17 queries. Of the remaining 26, 21 run within $3\times$ of plain DuckDB on the Graviton machine. Outlier Q30 is extremely aggregation-heavy: 90 sums that get transformed. }}
    \label{fig:clickbench}
    \vspace*{-7mm}
  \end{figure}

\subsection{Performance Impact (RQ2)}
\label{sec:analytical}

\vspace*{0mm}\noindent{\bf ClickBench.} We declare the \texttt{hits} table as the PU and mark \texttt{UserID} and \texttt{ClientIP} as protected columns, and then we run the ClickBench\footnote{\href{https://benchmark.clickhouse.com/}{https://benchmark.clickhouse.com/}} benchmark. Figure~\ref{fig:clickbench} shows the results on the Graviton4 machine: of the 43 queries, 20 of them incur less than $2.5\times$ overhead, while 17
are correctly rejected by the PAC checker (10 directly return protected columns, and 7 involve not implemented operators or unstable results). Since the PU is defined directly on the scanned table, DuckDB-PAC requires no PU-join; the measured overhead reflects only the cost of hashing PAC operators.

\begin{figure*}[ht]
  \centering
  \includegraphics[width=\linewidth]{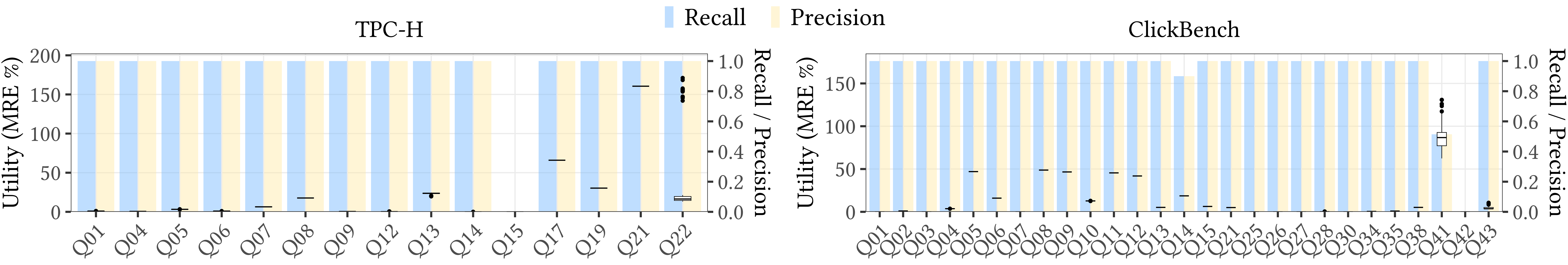}\vspace*{-4mm}
  \caption{\textbf TPC-H achieves a median MAPE of 3.2\% with perfect recall on 14 of 15 queries (Q15 returns zero recall because it joins on an exact equality between two independently noised aggregates. Similarly, ClickBench achieves a median of 3.7\% with perfect recall on 23 of 26 queries (Q42 performs \texttt{OFFSET 10000} on a noised ordering, yielding entirely different rows).}
  \label{fig:utility}
  \vspace*{-3mm}
\end{figure*}

\vspace*{2mm}\noindent{\bf TPC-H.} We run TPC-H at SF30 on the Apple M2 Max, declaring \texttt{customer} as the PU with five protected columns (\texttt{c\_custkey}, \texttt{c\_name}, \texttt{c\_address}, \texttt{c\_acctbal}, \texttt{c\_comment}). PAC links propagate the \emph{pu} hash and inject joins from lineitem $\to$ orders. Figure~\ref{fig:tpch} -- shown already on page 1 -- isolates their performance, showing that in most TPC-H queries which slow down, these (unavoidable) PU key joins are the main cause.\footnote{ClickBench has no joins, hence no PU key joins either.} Excluding the aggregation-heavy Q01 ($5.2\times$), average slowdown for the remaining queries is $1.6\times$. 
\newpage

\subsection{Quality (RQ3)}
\label{sec:utility}

\begin{figure}[t]
  \centering
  \includegraphics[width=\linewidth]{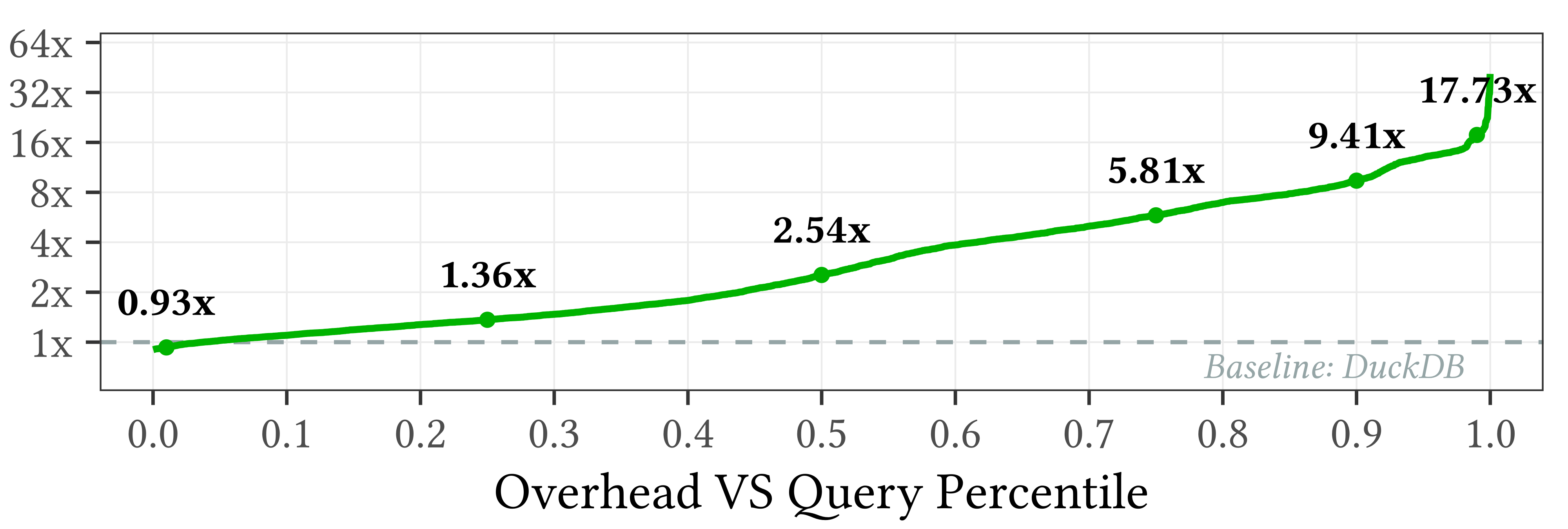}
  \vspace*{-8mm}
  \caption{\textbf{The median overhead of PAC is 2.54$\times$, and over 90\% of queries remain below 10$\times$ overhead.
}}
  \label{fig:sqlstorm}
  \vspace*{-3mm}
\end{figure}

\noindent\textbf{PacDiff.} To facilitate result quality evaluation, the PAC Rewriter supports \texttt{\bf set pac\_diffcols=$<$X$>$;}. If set, this will rewrite any query into a FULL OUTER JOIN  between the original and the PAC-privatized query, using the first $X$ columns of the result as equi-join key between them. On top sits a new \texttt{PacDiff} operator that substitutes all remaining (numeric) column values in the result with their Mean Absolute Percentage Error: \emph{abs(result - exact)/abs(exact)}. It outputs a \textbf{utility} metric which is the avg MAPE over all these columns.
\texttt{PacDiff} performs a Unix-style diff, classifying result rows present in both as "=", exact rows missing in the privatized result as "-", and rows only present in the privatized result as "+". We define \textbf{recall} as \emph{\#=rows/(\#=rows+\#-rows)} and \textbf{precision} as \emph{\#=rows/(\#=rows+\#+rows)}.
\texttt{pac\_diffcols} is a powerful tool to quickly assess the quality of the results of arbitrary privatized queries.

We first use this tool to investigate the effect of our lambda expression rewrites on utility.
In e.\:g., TPC-H Q08, a market share is computed by dividing the sum of sales in one region by the sum of sales. We  evaluate this division vector-lifted on the unnoised \texttt{pac\_sum()} result vectors containing the 64 outcomes for 64 possible worlds, into a vector of 64 market shares, which gets noised once:

{\small\begin{verbatim}
SELECT pac_noised(
       list_transform(
         list_zip(pac_sum(pu, CASE WHEN nation = 'BRAZIL'
                                   THEN sales ELSE 0 END)),
                   pac_sum(pu, sales), lambda x: x[1] / x[2]))
FROM ..compute sales from nation by year.. GROUP BY o_year;
\end{verbatim}}

The na\"ive counterpart would be to compute and noise both sums immediately with \texttt{pac\_noised\_sum()}, both yielding a single result, and dividing these two scalars. This approach noises twice, and also mixes outcomes from different worlds (increasing error).

{\small\begin{verbatim}
SELECT pac_noised_sum(pu, CASE WHEN nation = 'BRAZIL'
                               THEN sales ELSE 0 END)),
       / pac_noised_sum(pu, sales)
FROM ..compute sales from nation by year.. GROUP BY o_year;\end{verbatim}}

We micro-benchmark this rewrite by constructing 20 queries over TPC-H at SF1, each computing a grouped sum of $N$ ratio expressions of the form $100 \cdot \texttt{pac\_sum}(h, e_i) / \texttt{pac\_sum}(h, e)$, $h = \texttt{hash}(\texttt{o\_custkey})$. 
Figure~\ref{fig:ratio} shows that the lambda approach maintains consistently low error, as $N$ increases (10 runs).

Then, we execute each of all non-rejected TPC-H and ClickBench queries 100 times at SF30 and compare each PAC-privatized output against the corresponding non-private reference result. Figure~\ref{fig:utility} shows the per-query distribution of utility over the 100 runs with $mi = 1/128$: across both benchmarks, the majority of queries achieve perfect recall and precision, confirming that PAC noise does not eliminate or fabricate result groups: the median MAPE is 3.2\% for TPC-H and 3.7\% for ClickBench.


\begin{figure}[t]
  \centering\vspace*{-2mm}
  \includegraphics[width=\linewidth]{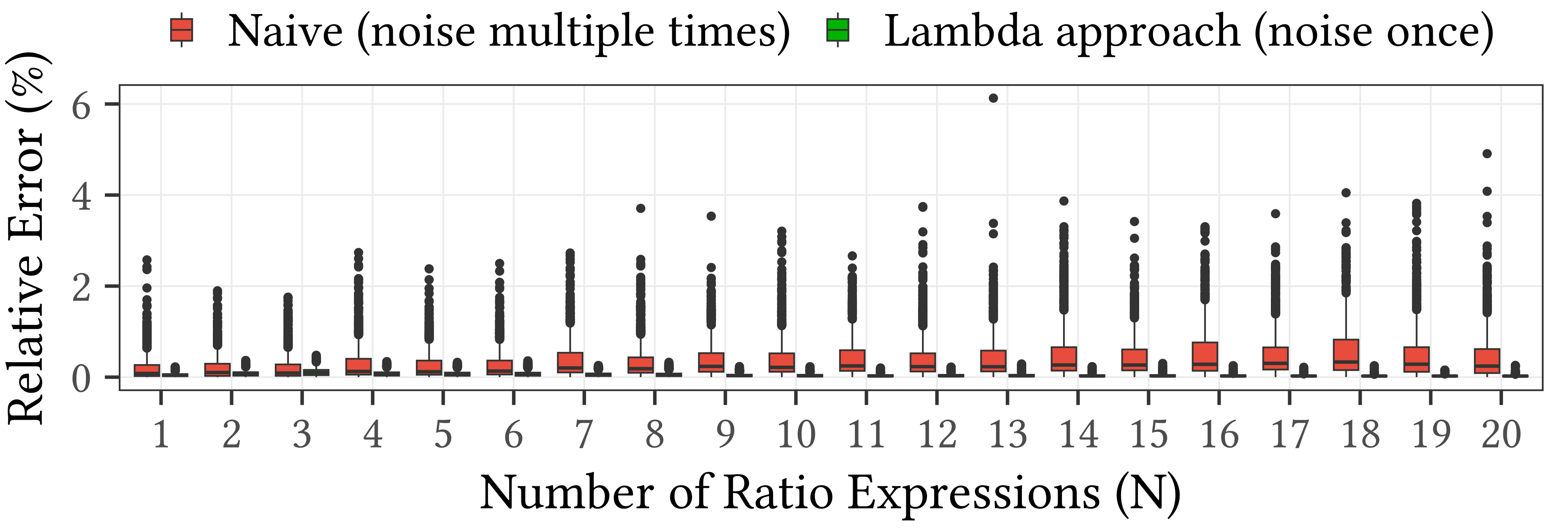}
  \vspace*{-8mm}
  \caption{\textbf{The na\"ive approach (red) independently noises each aggregate in an expression, causing growing error with more aggregates. Our lambda approach vector-lifts expressions using \texttt{list\_transform} (green) and then applies \texttt{pac\_noise()} only once, keeping error below $0.02\%$.}}
  \label{fig:ratio}
  \vspace*{-4mm}
\end{figure}

\subsection{Coverage (RQ4)} 
\label{sec:coverage}
To stress-test the PAC rewriter (RQ4), we use SQLStorm~\cite{sqlstorm2025} to generate 15301 of diverse DuckDB queries over the TPC-H schema (SF1, timeout 5s) and report the runtime slowdown in Figure~\ref{fig:sqlstorm}. Of all generated queries, $34.65\%$ are successfully PAC-rewritten and executed, and a further $14.92\%$ are correctly passed through unchanged because they reference no PU-linked table. An additional $6.94\%$ are correctly refused for attempting to release protected columns. The remaining queries exercise constructs not yet supported: window functions ($28.88\%$), recursive CTEs ($5.61\%$), and unsupported aggregate functions ($3.39\%$, including \texttt{string\_agg}, \texttt{stddev}, \texttt{first}). This is the low-hanging fruit in SIMD-PAC-DB to increase the percentage of privatizable queries to $>$80\%.



\section{Related Work}
\label{sec:related}
\vspace*{0mm}\noindent{\bf Differential Privacy.}
Differential Privacy (DP)~\cite{dwork2006calibrating, dwork2008survey, dwork2014foundations} guarantees that query outputs change negligibly when one individual's data is added or removed, by adding noise calibrated to query \emph{sensitivity} and a privacy budget $\epsilon$. The composition theorem~\cite{kairouz2015composition} bounds the total privacy cost of multiple releases. While strong and compositional, DP is costly in practice: sensitivity must be known or conservatively bounded in advance, worst-case noise calibration degrades utility on complex analytical queries, and integrating DP into general-purpose engines requires expert reasoning about sensitivity across joins, groupings, and nested queries.

\vspace*{2mm}\noindent{\bf PAC Privacy.}
PAC privacy~\cite{hanshen2023crypto} bypasses the need for worst-case sensitivity bounds by measuring query stability via \textit{probabilistic subsampling}.
This makes PAC privacy significantly more accurate for analytical workloads without requiring the analyst to know sensitivity in advance~\cite{sridhar2026pacdb}. 
While the original PAC-DB independently resampled the secret subset for every scalar cell in the query output, yielding only \emph{cell-level} membership privacy, we leverage a recent adaptive composition theorem~\cite{zhu2026pacadversarial} to maintain a persistent secret subset across all releases within a query. This ensures stronger provable guarantees against \emph{query-level} membership inference attacks, where an adversary attempts to infer if a PU was \emph{ever} used to answer a query.
Our work is the first to compile these advanced PAC semantics directly into a database engine's operator model.

\vspace*{2mm}
\noindent{\bf Differentially Private Database Systems.}
Several systems integrate DP into general-purpose query processing, tackling different aspects of the sensitivity problem. PINQ~\cite{mcsherry2009pinq} introduced operator-level sensitivity tracking; wPINQ~\cite{proserpio2014wpinq} extended it to weighted datasets. 
GUPT~\cite{mohan2012gupt} avoids explicit sensitivity computation via sample-and-aggregate, similar to PAC, but offers weaker utility on complex queries because DP does not exploit data-dependent stability~\cite{zhu2026pacadversarial}.
FLEX~\cite{johnson2018flex} introduced elastic sensitivity for DP count queries with joins; CHORUS~\cite{johnson2020chorus} extended this to sums and averages via contribution clipping, though functions requiring the exponential mechanism still demand manual analyst effort. PrivateSQL~\cite{kotsogiannis2019privatesql} releases DP synopsis views that can answer an unlimited number of subsequent queries, but is confined to linear queries and requires a representative workload to be specified a priori. Wilson et al.~\cite{wilson2020dpbounded} address bounded user contributions directly in SQL; Yuan et al.~\cite{yuan2012lowrank} optimize batch queries under DP via low-rank mechanisms.
Beyond single-database settings, Shrinkwrap~\cite{bater2018shrinkwrap} reduces intermediate result sizes in differentially private data federations, while IncShrink~\cite{wang2022incshrink} combines incremental MPC with DP for outsourced databases. On the output side, HDPView~\cite{kato2022hdpview} releases DP materialized views for high-dimensional relational data, and DPXPlain~\cite{tao2022dpxplain} privately explains aggregate query answers to analysts. In dynamic settings, Dong et al.~\cite{dong2024continual,dong2023multiple} study continual observation of joins and answering multiple relational queries under DP, while Zhang et al.~\cite{zhang2024dpstream} demonstrate DP stream processing at scale. Fang et al.~\cite{fang2022shiftedinverse} propose a general mechanism for monotonic functions under user-level privacy; the $\epsilon$KTELO framework~\cite{zhang2020ektelo} composes DP computations modularly but still requires the analyst to reason about sensitivity.
Orthogonal to central DP, local DP~\cite{wang2024pripltree,localDP2020} and shuffle DP~\cite{bittau2017prochlo,erlingsson2019amplification} eliminate the need for a trusted curator by having contributors randomize their data before submission, at the cost of significantly higher noise. Cormode et al.~\cite{cormode2018marginal} study marginal release under local DP, while recent work pushes these models further: Luo et al.~\cite{luo2025rm2} address counting queries under shuffle DP, Xie et al.~\cite{xie2025cardinality} tackle DP cardinality estimation in streams, and He et al.~\cite{he2025triangle} consider triangle counting under edge local DP.

\newpage

\vspace*{-2mm}
\noindent{\bf Monte Carlo and Approximate Query Processing.}
MCDB~\cite{jampani2011mcdb} evaluates queries across multiple stochastic database instances using ``tuple-bundle'' processing, executing once across all Monte Carlo iterations simultaneously. SimSQL~\cite{simsql} extends this idea into a full in-database statistical computing system, compiling Monte Carlo-style repeated evaluation into efficient single-pass relational execution. Both are conceptually related to our SIMD PAC operators, which evaluate 64 subsamples in a single pass using bit-sliced hashes, but our approach compresses subsample membership into a 64-bit hash per PU rather than carrying full attribute arrays. Quickr~\cite{quickr} and BlinkDB~\cite{blinkdb} pursue approximate query answering with bounded error, trading accuracy for latency --- a complementary but orthogonal goal to privacy.

\vspace*{2mm}
\noindent{\bf SIMD Query Processing.}
Exploiting SIMD for query processing has a long history,  in data decompression~\cite{willhalm2009simdscan, simd-compression}, index search~\cite{simd-search}, bloom-filter lookup~\cite{simd-bf}, filter predicates~\cite{simd-filter} and regular expressions~\cite{simd-regexp}.
Polychroniou et al.~\cite{polychroniou2015simd,simd-towards} survey  vectorization for query processing and propose the VIP system~\cite{simd-db}.
SIMD-PAC-DB avoids using explicit SIMD instructions, instead relying on auto-vectorizing C++ compilers, for better portability and future-proofness -- similar to the recent FastLanes compression encodings~\cite{Afroozeh2023}.
Accelerating aggregate updates with SIMD had not been achieved before, likely because of the combination of random memory access to reach an aggregate state and limited computational work required.
The need to update 64 worlds with the same value make PAC aggregates very suitable for SIMD acceleration.



\section{Conclusion and Future Work}
SIMD-PAC-DB strongly improves ($>$10x) the performance of the recently proposed PAC-DB system -- which enabled black box privatized execution of arbitrary queries. We think that the combination of automatic privatization of arbitrarily complex queries, with fast query execution significantly advances the ease of use and functionality of privacy-aware data management in practice. 

SIMD-PAC-DB executes a query in a single pass thanks to new bit-parallel stochastic aggregate functions, for which we contribute SIMD-friendly optimizations. We formally prove that the SIMD-PAC-DB query rewriter produces plans semantically equivalent to PAC-DB, preserving its privacy guarantees and utility. Our evaluation shows good performance and utility on thousands of automatically privatized TPC-H, ClickBench and SQLStorm queries. We release DuckDB-PAC in open-source as a DuckDB v1.5 extension that is easy to get started with (\texttt{install pac from community;}).

Future work includes (i) research into the support for complex relational operators (e.g. window functions), (ii) further research into PAC-DB semantics regarding e.g. row filtering (categorical rather than continuous noising) as well as fine-grained noising methods for these (iii) research into supporting multiple privacy units (PUs) as well as cyclical join query shapes and recursion. Together, these directions suggest the opening of a new research agenda, that aims to deeply intertwine PAC privacy and data systems.
\begin{acks}
 We thank Mayuri Sridhar, Srini Devadas, Michael Noguera and Xiangyao Yu for inspirational conversations that led to this paper.
\end{acks}
\clearpage

\bibliographystyle{ACM-Reference-Format}
\bibliography{sample}

@misc{sridhar2026pacdb,
      author = {Mayuri Sridhar and Michael A. Noguera and Chaitanyasuma Jain and Kevin Kristensen and Srinivas Devadas and Hanshen Xiao and Xiangyao Yu},
      title = {{PAC}-Private Databases},
      howpublished = {Cryptology {ePrint} Archive, Paper 2026/238},
      year = {2026},
      url = {https://eprint.iacr.org/2026/238}
}

@INPROCEEDINGS{sridhar2025pac,
  author    = {Sridhar, Mayuri and Xiao, Hanshen and Devadas, Srinivas},
  booktitle = {2025 IEEE Symposium on Security and Privacy (SP)},
  title     = {{PAC-Private Algorithms}},
  year      = {2025},
  pages     = {3839-3857},
  doi       = {10.1109/SP61157.2025.00034},
  publisher = {IEEE Computer Society},
  address   = {Los Alamitos, CA, USA},
  month     = May
}

@misc{zhu2026pacadversarial,
  title        = {PAC-Private Responses with Adversarial Composition},
  author       = {Xiaochen Zhu and Mayuri Sridhar and Srinivas Devadas},
  year         = {2026},
  eprint       = {2601.14033},
  archivePrefix= {arXiv},
  primaryClass = {cs.LG},
  url          = {https://arxiv.org/abs/2601.14033},
}

@article{dwork2006calibrating,
  title   = {Calibrating Noise to Sensitivity in Private Data Analysis},
  author  = {Dwork, Cynthia and McSherry, Frank and Nissim, Kobbi and Smith, Adam},
  journal = {Theory of Cryptography Conference (TCC)},
  year    = {2006}
}

@article{dwork2008survey,
  title   = {Differential Privacy: A Survey of Results},
  author  = {Dwork, Cynthia},
  journal = {Proceedings of the 5th International Conference on Theory and Applications of Models of Computation},
  year    = {2008}
}

@article{dwork2014foundations,
  author  = {Cynthia Dwork and Aaron Roth},
  title   = {The Algorithmic Foundations of Differential Privacy},
  journal = {Found. Trends Theor. Comput. Sci.},
  volume  = {9},
  number  = {3--4},
  pages   = {211--407},
  year    = {2014},
  doi     = {10.1561/0400000042}
}

@inproceedings{kairouz2015composition,
  author    = {Kairouz, Peter and Oh, Sewoong and Viswanath, Pramod},
  title     = {The composition theorem for differential privacy},
  booktitle = {Proceedings of the 32nd International Conference on Machine Learning},
  pages     = {1376--1385},
  year      = {2015},
  series    = {ICML'15}
}

@inproceedings{mcsherry2009pinq,
  title     = {Privacy Integrated Queries: An Extensible Platform for Privacy-Preserving Data Analysis},
  author    = {McSherry, Frank},
  booktitle = {SIGMOD},
  year      = {2009}
}

@article{proserpio2014wpinq,
  author  = {Davide Proserpio and Sharon Goldberg and Frank McSherry},
  title   = {Calibrating Data to Sensitivity in Private Data Analysis: A Platform for Differentially-Private Analysis of Weighted Datasets},
  journal = {Proc. VLDB Endow.},
  volume  = {7},
  number  = {8},
  pages   = {637--648},
  year    = {2014}
}

@inproceedings{mohan2012gupt,
  author    = {Prashanth Mohan and Abhradeep Thakurta and Elaine Shi and Dawn Song and David E. Culler},
  title     = {{GUPT}: Privacy Preserving Data Analysis Made Easy},
  booktitle = {Proceedings of the ACM SIGMOD International Conference on Management of Data},
  year      = {2012},
  pages     = {349--360},
  publisher = {ACM}
}

@article{johnson2018flex,
  author  = {Noah Johnson and Joseph P. Near and Dawn Song},
  title   = {Towards Practical Differential Privacy for {SQL} Queries},
  journal = {Proc. VLDB Endow.},
  volume  = {11},
  number  = {5},
  pages   = {526--539},
  year    = {2018}
}

@inproceedings{johnson2020chorus,
  author    = {Noah Johnson and Joseph P. Near and Joseph M. Hellerstein and Dawn Song},
  title     = {{CHORUS}: A Programming Framework for Building Scalable Differential Privacy Mechanisms},
  booktitle = {2020 IEEE European Symposium on Security and Privacy (EuroS\&P)},
  year      = {2020},
  pages     = {535--551}
}

@article{kotsogiannis2019privatesql,
  author  = {Ios Kotsogiannis and Yuchao Tao and Xi He and Maryam Fanaeepour and Ashwin Machanavajjhala and Michael Hay and Gerome Miklau},
  title   = {{PrivateSQL}: A Differentially Private {SQL} Query Engine},
  journal = {Proc. VLDB Endow.},
  volume  = {12},
  number  = {11},
  pages   = {1371--1384},
  year    = {2019}
}

@article{wilson2020dpbounded,
  author  = {Royce Wilson and Celia Zhang and William Lam and Damien Desfontaines and Daniel Simmons-Marengo and Bryant Gipson},
  title   = {Differentially Private {SQL} with Bounded User Contribution},
  journal = {Proceedings on Privacy Enhancing Technologies},
  volume  = {2020},
  pages   = {230--250},
  year    = {2020}
}

@article{dong2024continual,
  author    = {Wei Dong and Zijun Chen and Qiyao Luo and Elaine Shi and Ke Yi},
  title     = {Continual Observation of Joins under Differential Privacy},
  journal   = {Proc. ACM Manag. Data},
  volume    = {2},
  number    = {3},
  articleno = {128},
  year      = {2024}
}

@article{dong2023multiple,
  author    = {Wei Dong and Dajun Sun and Ke Yi},
  title     = {Better than Composition: How to Answer Multiple Relational Queries under Differential Privacy},
  journal   = {Proc. ACM Manag. Data},
  volume    = {1},
  number    = {2},
  articleno = {123},
  year      = {2023}
}

@article{zhang2020ektelo,
  author    = {Dan Zhang and Ryan McKenna and Ios Kotsogiannis and George Bissias and Michael Hay and Ashwin Machanavajjhala and Gerome Miklau},
  title     = {$\epsilon${KTELO}: A Framework for Defining Differentially Private Computations},
  journal   = {ACM Trans. Database Syst.},
  volume    = {45},
  number    = {1},
  articleno = {2},
  year      = {2020}
}

@article{zhang2024dpstream,
  author  = {Bing Zhang and Vadym Doroshenko and Peter Kairouz and Thomas Steinke and Abhradeep Thakurta and Ziyin Ma and Eidan Cohen and Himani Apte and Jodi Spacek},
  title   = {Differentially Private Stream Processing at Scale},
  journal = {Proc. VLDB Endow.},
  volume  = {17},
  number  = {12},
  pages   = {4145--4158},
  year    = {2024},
  doi     = {10.14778/3685800.3685833}
}

@article{tao2022dpxplain,
  author  = {Yuchao Tao and Amir Gilad and Ashwin Machanavajjhala and Sudeepa Roy},
  title   = {DPXPlain: privately explaining aggregate query answers},
  journal = {Proc. VLDB Endow.},
  volume  = {16},
  number  = {1},
  pages   = {113--126},
  year    = {2022}
}

@article{kato2022hdpview,
  author  = {Fumiyuki Kato and Tsubasa Takahashi and Shun Takagi and Yang Cao and Seng Pei Liew and Masatoshi Yoshikawa},
  title   = {HDPView: differentially private materialized view for exploring high dimensional relational data},
  journal = {Proc. VLDB Endow.},
  volume  = {15},
  number  = {9},
  pages   = {1766--1778},
  year    = {2022}
}

@article{wang2024pripltree,
  author  = {Leixia Wang and Qingqing Ye and Haibo Hu and Xiaofeng Meng},
  title   = {{PriPL-Tree}: Accurate Range Query for Arbitrary Distribution under Local Differential Privacy},
  journal = {Proc. VLDB Endow.},
  volume  = {17},
  number  = {11},
  pages   = {3031--3044},
  year    = {2024}
}

@inproceedings{fang2022shiftedinverse,
  author    = {Juanru Fang and Wei Dong and Ke Yi},
  title     = {Shifted Inverse: A General Mechanism for Monotonic Functions under User Differential Privacy},
  booktitle = {Proceedings of the 2022 ACM SIGSAC Conference on Computer and Communications Security},
  pages     = {1009--1022},
  year      = {2022},
  series    = {CCS '22}
}

@article{yuan2012lowrank,
  author  = {Ganzhao Yuan and Zhenjie Zhang and Marianne Winslett and Xiaokui Xiao and Yin Yang and Zhifeng Hao},
  title   = {Low-rank mechanism: optimizing batch queries under differential privacy},
  journal = {Proc. VLDB Endow.},
  volume  = {5},
  number  = {11},
  pages   = {1352--1363},
  year    = {2012}
}

@article{cormode2018marginal,
  author    = {Graham Cormode and Tejas Kulkarni and Divesh Srivastava},
  title     = {Marginal Release Under Local Differential Privacy},
  booktitle = {Proceedings of the 2018 International Conference on Management of Data (SIGMOD)},
  year      = {2018},
  pages     = {131--146}
}

@article{luo2025rm2,
  author    = {Qiyao Luo and Jianzhe Yu and Wei Dong and Quanqing Xu and Chuanhui Yang and Ke Yi},
  title     = {RM2: Answer Counting Queries Efficiently under Shuffle Differential Privacy},
  journal   = {Proc. ACM Manag. Data},
  volume    = {3},
  number    = {3},
  articleno = {210},
  year      = {2025},
  doi       = {10.1145/3725415}
}

@inproceedings{kuiperaggr,
  author       = {Laurens Kuiper and
                  Peter Boncz and
                  Hannes M{\"{u}}hleisen},
  title        = {Robust External Hash Aggregation in the Solid State Age},
  booktitle    = {{ICDE}},
  pages        = {3753--3766},
  publisher    = {{IEEE}},
  year         = {2024}
}

@inproceedings{xie2025cardinality,
  author    = {Dongdong Xie and Pinghui Wang and Quanqing Xu and Chuanhui Yang and Rundong Li},
  title     = {Efficient and Accurate Differentially Private Cardinality Continual Releases},
  journal   = {Proc. ACM Manag. Data},
  volume    = {3},
  number    = {3},
  articleno = {151},
  year      = {2025},
  doi       = {10.1145/3725288}
}

@inproceedings{he2025triangle,
  author    = {Yizhang He and Kai Wang and Wenjie Zhang and Xuemin Lin and Ying Zhang and Wei Ni},
  title     = {Robust Privacy-Preserving Triangle Counting under Edge Local Differential Privacy},
  journal   = {Proc. ACM Manag. Data},
  volume    = {3},
  number    = {3},
  articleno = {211},
  year      = {2025},
  doi       = {10.1145/3725348}
}

@article{localDP2020,
  author  = {Xingxing Xiong and Shubo Liu and Dan Li and Zhaohui Cai and Xiaoguang Niu},
  title   = {A Comprehensive Survey on Local Differential Privacy},
  journal = {Security and Communication Networks},
  volume  = {2020},
  pages   = {8829523},
  year    = {2020},
  doi     = {10.1155/2020/8829523}
}

@inproceedings{bittau2017prochlo,
  author    = {Andrea Bittau and others},
  title     = {Prochlo: Strong Privacy for Analytics in the Crowd},
  booktitle = {Proceedings of the 26th Symposium on Operating Systems Principles (SOSP)},
  year      = {2017},
  pages     = {441--459},
  publisher = {ACM}
}

@inproceedings{erlingsson2019amplification,
  author    = {\'Ulfar Erlingsson and Vitaly Feldman and Ilya Mironov and Ananth Raghunathan and Kunal Talwar and Abhradeep Thakurta},
  title     = {Amplification by Shuffling: From Local to Central Differential Privacy via Anonymity},
  booktitle = {Proceedings of the 2019 ACM-SIAM Symposium on Discrete Algorithms (SODA)},
  year      = {2019},
  pages     = {2468--2479}
}

@article{bater2018shrinkwrap,
  author  = {Johes Bater and Xi He and William Ehrich and Ashwin Machanavajjhala and Jennie Rogers},
  title   = {Shrinkwrap: efficient SQL query processing in differentially private data federations},
  journal = {Proc. VLDB Endow.},
  volume  = {12},
  number  = {3},
  pages   = {307--320},
  year    = {2018}
}

@inproceedings{wang2022incshrink,
  author    = {Chenghong Wang and Johes Bater and Kartik Nayak and Ashwin Machanavajjhala},
  title     = {IncShrink: Architecting Efficient Outsourced Databases using Incremental MPC and Differential Privacy},
  booktitle = {Proceedings of the 2022 International Conference on Management of Data (SIGMOD)},
  year      = {2022},
  pages     = {818--832},
  doi       = {10.1145/3514221.3526151}
}

@article{jampani2011mcdb,
  author    = {Ravi Jampani and Fei Xu and Mingxi Wu and Luis Perez and Chris Jermaine and Peter J. Haas},
  title     = {The Monte Carlo database system: Stochastic analysis close to the data},
  journal   = {ACM Trans. Database Syst.},
  volume    = {36},
  number    = {3},
  articleno = {18},
  year      = {2011},
  doi       = {10.1145/2000824.2000828}
}

@inproceedings{raasveldt2019duckdb,
  author    = {Mark Raasveldt and Hannes M{\"u}hleisen},
  title     = {DuckDB: an Embeddable Analytical Database},
  booktitle = {Proceedings of the 2019 International Conference on Management of Data (SIGMOD)},
  year      = {2019},
  pages     = {1981--1984},
  publisher = {{ACM}}
}

@inproceedings{raasveldt2020embedded,
  author    = {Mark Raasveldt and Hannes M{\"u}hleisen},
  title     = {Data Management for Data Science - Towards Embedded Analytics},
  booktitle = {10th Conference on Innovative Data Systems Research (CIDR 2020)},
  year      = {2020},
  publisher = {www.cidrdb.org}
}

@inproceedings{polychroniou2015simd,
  author    = {Orestis Polychroniou and Arun Raghavan and Kenneth A. Ross},
  title     = {Rethinking SIMD Vectorization for In-Memory Databases},
  booktitle = {Proceedings of the 2015 ACM SIGMOD International Conference on Management of Data},
  year      = {2015},
  pages     = {1493--1508},
  doi       = {10.1145/2723372.2747645}
}

@article{willhalm2009simdscan,
  author  = {Thomas Willhalm and Nicolae Popovici and Yazan Boshmaf and Hasso Plattner and Alexander Zeier and Jan Schaffner},
  title   = {SIMD-scan: ultra fast in-memory table scan using on-chip vector processing units},
  journal = {Proc. VLDB Endow.},
  volume  = {2},
  number  = {1},
  pages   = {385--394},
  year    = {2009},
  doi     = {10.14778/1687627.1687671}
}

@misc{appleby2008murmurhash,
  author       = {Austin Appleby},
  title        = {{MurmurHash}},
  year         = {2008},
  howpublished = {\url{https://github.com/aappleby/smhasher}}
}

@article{sqlstorm2025,
  author  = {Tobias Schmidt and Viktor Leis and Peter Boncz and Thomas Neumann},
  title   = {SQLStorm: Taking Database Benchmarking into the LLM Era},
  journal = {Proc. {VLDB} Endow.},
  volume  = {18},
  number  = {11},
  pages   = {4144--4157},
  year    = {2025}
}

@inproceedings{blinkdb,
author = {Agarwal, Sameer and Mozafari, Barzan and Panda, Aurojit and Milner, Henry and Madden, Samuel and Stoica, Ion},
title = {BlinkDB: queries with bounded errors and bounded response times on very large data},
year = {2013},
isbn = {9781450319942},
publisher = {Association for Computing Machinery},
address = {New York, NY, USA},
url = {https://doi.org/10.1145/2465351.2465355},
doi = {10.1145/2465351.2465355},
booktitle = {Proceedings of the 8th ACM European Conference on Computer Systems},
pages = {29–42},
numpages = {14},
location = {Prague, Czech Republic},
series = {EuroSys '13}
}

@inproceedings{quickr,
author = {Kandula, Srikanth et al.},
title = {Quickr: Lazily Approximating Complex AdHoc Queries in BigData Clusters},
year = {2016},
isbn = {9781450335317},
publisher = {Association for Computing Machinery},
address = {New York, NY, USA},
url = {https://doi.org/10.1145/2882903.2882940},
doi = {10.1145/2882903.2882940},
booktitle = {Proceedings of the 2016 International Conference on Management of Data},
pages = {631–646},
numpages = {16},
location = {San Francisco, California, USA},
series = {SIGMOD '16}
}

@inproceedings{hanshen2023crypto,
  author    = {Hanshen Xiao and
               Srinivas Devadas},
  title     = {{PAC} Privacy: Automatic Privacy Measurement and Control of Data Processing},
  booktitle = {Advances in Cryptology – CRYPTO 2023},
  volume    = {14082},
  pages     = {611--644},
  year      = {2023}
}

@inproceedings{carlini2022membership,
  title={Membership inference attacks from first principles},
  author={Carlini, Nicholas and Chien, Steve and Nasr, Milad and Song, Shuang and Terzis, Andreas and Tramer, Florian},
  booktitle={2022 IEEE symposium on security and privacy (SP)},
  pages={1897--1914},
  year={2022},
  organization={IEEE}
}

@inproceedings{simsql,
author = {Cai, Zhuhua and Vagena, Zografoula and Perez, Luis and Arumugam, Subramanian and Haas, Peter J. and Jermaine, Christopher},
title = {Simulation of database-valued markov chains using SimSQL},
year = {2013},
isbn = {9781450320375},
publisher = {Association for Computing Machinery},
address = {New York, NY, USA},
url = {https://doi.org/10.1145/2463676.2465283},
doi = {10.1145/2463676.2465283},
booktitle = {Proceedings of the 2013 ACM SIGMOD International Conference on Management of Data},
pages = {637–648},
numpages = {12},
location = {New York, New York, USA},
series = {SIGMOD '13}
}

@inproceedings{simd-compression,
  author       = {Benjamin Schlegel and
                  Rainer Gemulla and
                  Wolfgang Lehner},
   title        = {Fast integer compression using {SIMD} instructions},
  booktitle    = {{DaMoN}},
    pages        = {34--40},
  publisher    = {{ACM}},
  year         = {2010}
}

@inproceedings{simd-search,
  author       = {Benjamin Schlegel and
                  Rainer Gemulla and
                  Wolfgang Lehner},
  title        = {k-ary search on modern processors},
  booktitle    = {{DaMoN}},
  pages        = {52--60},
  publisher    = {{ACM}},
  year         = {2009}
  }

@inproceedings{simd-bf,
  author       = {Orestis Polychroniou and
                  Kenneth A. Ross},
  editor       = {Alfons Kemper and
                  Ippokratis Pandis},
  title        = {Vectorized Bloom filters for advanced {SIMD} processors},
  booktitle    = {{DaMoN}},
  pages        = {6:1--6:6},
  publisher    = {{ACM}},
  year         = {2014}
}

@inproceedings{simd-regexp,
  author       = {Evangelia A. Sitaridi and
                  Orestis Polychroniou and
                  Kenneth A. Ross},
  title        = {SIMD-accelerated regular expression matching},
  booktitle    = {{DaMoN}},
  pages        = {8:1--8:7},
  publisher    = {{ACM}},
  year         = {2016}
  }

@inproceedings{simd-filter,
  author       = {Hao Jiang and
                  Aaron J. Elmore},
  title        = {Boosting data filtering on columnar encoding with {SIMD}},
  booktitle    = {{DaMoN}},
   pages        = {6:1--6:10},
  publisher    = {{ACM}},
  year         = {2018}
  }

@article{Afroozeh2023,
    author = {Afroozeh, Azim and Boncz, Peter},
    title = {The FastLanes Compression Layout: Decoding $>$ 100 Billion Integers per Second with Scalar Code},
    year = {2023},
    publisher = {VLDB Endowment},
    volume = {16},
    number = {9},
    journal = {Proc. VLDB Endow.},
    month = {jul},
    pages = {2132–2144},
    numpages = {13}
}

@inproceedings{simd-towards,
  author       = {Orestis Polychroniou and
                  Kenneth A. Ross},
  title        = {Towards Practical Vectorized Analytical Query Engines},
  booktitle    = {{DaMoN}},
    pages        = {10:1--10:7},
  publisher    = {{ACM}},
  year         = {2019}
  }

@article{simd-db,
  author       = {Orestis Polychroniou and
                  Kenneth A. Ross},
  title        = {{VIP:} {A} {SIMD} vectorized analytical query engine},
  journal      = {{VLDB} J.},
  volume       = {29},
  number       = {6},
  pages        = {1243--1261},
  year         = {2020}
}

\end{document}